\documentclass[11pt]{article}

\usepackage[authoryear]{natbib}
\usepackage{bbm}
\usepackage{framed} 
\usepackage{url} 
\usepackage[skins]{tcolorbox}
\usepackage{complexity}
\usepackage{booktabs}
\usepackage{amsmath,amssymb}
\usepackage{amsthm}
\usepackage{thmtools} 
\usepackage{thm-restate}
\usepackage{nicefrac}
\usepackage{csquotes}
\usepackage{microtype}
\usepackage{calc}
\usepackage{enumerate}
\usepackage{enumitem}
\usepackage{textgreek}
\usepackage{setspace}
\usepackage{parskip}
\usepackage[hidelinks]{hyperref}

\usepackage{fullpage}
\allowdisplaybreaks

\usepackage{graphicx}
\usepackage[font=footnotesize,labelfont=bf]{subcaption}
\usepackage[font=footnotesize,labelfont=bf]{caption}
\usepackage[nobreak=true,skipabove=7pt]{mdframed}
\usepackage{appendix}
\usepackage[colorinlistoftodos]{todonotes}

\usepackage{algorithm}
\usepackage{algcompatible}
\usepackage[noend]{algpseudocode}

\usepackage{xr}
\usepackage{array}
\usepackage{xspace}

\usepackage[capitalise]{cleveref}
\crefrangelabelformat{enumi}{#3#1#4--#5#2#6}
\usepackage{chngcntr}
\usepackage{mathtools}

\crefname{claim}{Claim}{Claims}

\usepackage{nameref}
\usepackage{tikz}
\usetikzlibrary{patterns}

\DeclareMathOperator*{\expectation}{\mathbb{E}}

\DeclareMathOperator*{\probability}{\mathbb{P}}

\newcommand\eps{\epsilon}
\renewcommand\epsilon{\varepsilon}

\newcommand{\prob}{\probability\probarg}
\DeclarePairedDelimiterX{\probarg}[1]{(}{)}{%
	\ifnum\currentgrouptype=16 \else\begingroup\fi
	\activatebar#1
	\ifnum\currentgrouptype=16 \else\endgroup\fi
}

\newcommand{\expect}{\expectation\expectarg}
\DeclarePairedDelimiterX{\expectarg}[1]{[}{]}{%
	\ifnum\currentgrouptype=16 \else\begingroup\fi
	\activatebar#1
	\ifnum\currentgrouptype=16 \else\endgroup\fi
}

\newcommand{\expectover}[1]{\expectation_{#1}\expectarg}

\newcommand{\innermid}{\nonscript\;\delimsize\vert\nonscript\;}
\newcommand{\activatebar}{%
	\begingroup\lccode`\~=`\|
	\lowercase{\endgroup\let~}\innermid 
	\mathcode`|=\string"8000
}

\newcommand{\innermids}{\nonscript\;\delimsize\vert\delimsize\vert\nonscript\;}
\newcommand{\activatebars}{%
	\begingroup\lccode`\~=`\|
	\lowercase{\endgroup\let~}\innermids 
	\mathcode`|=\string"8000
}

\newcommand\opt{\textsc{Opt}\xspace}

\newcommand\alg{\textsc{Alg}\xspace}

\newcommand\defeq{ \xspace \stackrel{\tiny \mathclap{\mbox{(def)}}}{=} \xspace}

\counterwithin{equation}{section}

\usepackage{eqparbox}

\theoremstyle{plain}

\newtheorem{theorem}{Theorem}[section]

\newtheorem{lemma}[theorem]{Lemma}

\newtheorem{observation}[theorem]{Observation}
\newtheorem{claim}[theorem]{Claim}

\newtheorem{corollary}[theorem]{Corollary}

\newtheorem{assumption}[theorem]{Assumption}

\newlength{\continueindent}
\usepackage{etoolbox}
\makeatletter
\newcommand*{\ALG@customparshape}{\parshape 2 \leftmargin \linewidth \dimexpr\ALG@tlm+\continueindent\relax \dimexpr\linewidth+\leftmargin-\ALG@tlm-\continueindent\relax}
\apptocmd{\ALG@beginblock}{\ALG@customparshape}{}{\errmessage{failed to patch}}
\makeatother

\makeatletter
\def\thm@space@setup{%
	\thm@preskip=\parskip \thm@postskip=0pt
}
\makeatother

\usepackage{etoolbox}
\usepackage{tikz}
\usetikzlibrary{tikzmark}
\usetikzlibrary{calc}

\errorcontextlines\maxdimen

\newcommand{\ALGtikzmarkcolor}{black}% customise this, if you want
\newcommand{\ALGtikzmarkextraindent}{4pt}% customise this, if you want
\newcommand{\ALGtikzmarkverticaloffsetstart}{-.5ex}% customise this, if you want
\newcommand{\ALGtikzmarkverticaloffsetend}{-.5ex}% customise this, if you want
\makeatletter
\newcounter{ALG@tikzmark@tempcnta}

\newcommand\ALG@tikzmark@start{%
	\global\let\ALG@tikzmark@last\ALG@tikzmark@starttext%
	\expandafter\edef\csname ALG@tikzmark@\theALG@nested\endcsname{\theALG@tikzmark@tempcnta}%
	\tikzmark{ALG@tikzmark@start@\csname ALG@tikzmark@\theALG@nested\endcsname}%
	\addtocounter{ALG@tikzmark@tempcnta}{1}%
}

\def\ALG@tikzmark@starttext{start}
\newcommand\ALG@tikzmark@end{%
	\ifx\ALG@tikzmark@last\ALG@tikzmark@starttext
	\else
	\tikzmark{ALG@tikzmark@end@\csname ALG@tikzmark@\theALG@nested\endcsname}%
	\tikz[overlay,remember picture] \draw[\ALGtikzmarkcolor] let \p{S}=($(pic cs:ALG@tikzmark@start@\csname ALG@tikzmark@\theALG@nested\endcsname)+(\ALGtikzmarkextraindent,\ALGtikzmarkverticaloffsetstart)$), \p{E}=($(pic cs:ALG@tikzmark@end@\csname ALG@tikzmark@\theALG@nested\endcsname)+(\ALGtikzmarkextraindent,\ALGtikzmarkverticaloffsetend)$) in (\x{S},\y{S})--(\x{S},\y{E});%
	\fi
	\gdef\ALG@tikzmark@last{end}%
}

\apptocmd{\ALG@beginblock}{\ALG@tikzmark@start}{}{\errmessage{failed to patch}}
\pretocmd{\ALG@endblock}{\ALG@tikzmark@end}{}{\errmessage{failed to patch}}
\makeatother
\algblock[with]{With}{EndWith}
\algblockdefx[With]{With}{EndWith}%
[1]{\textbf{with} #1 \textbf{do}}%
{}

\makeatletter
\ifthenelse{\equal{\ALG@noend}{t}}%
{\algtext*{EndWith}}
{}%
\makeatother

\usepackage{tikz}
\usetikzlibrary{arrows.meta}

\renewcommand{\Pr}[1]{\probability\left[#1\right]}
\newcommand{\Ex}[1]{\expectation\left[#1\right]}

\newcommand\given{\;\middle\vert\;}

\newcommand\bbR{\mathbb{R}}
\newcommand\bbN{\mathbb{N}}
\newcommand\ty{\tilde{y}}
\newcommand\rts{\tilde{r}^*}
\newcommand\rc{\check{r}}
\newcommand\cD{\mathcal{D}}
\newcommand\cDp{\mathcal{D}'}

\def \bX {\mathbf{X}}
\def \bT {\mathbf{T}}

\newcommand{\MAX}{\mathsf{MAX}}
\newcommand{\OPT}{\mathsf{OPT}}

\newcommand{\ALG}{\mathsf{ALG}}
\newcommand{\threephase}{\textsc{ThreePhase}}

\newcommand{\CDF}{\mathsf{CDF}}

\title{Prophet and Secretary at the Same Time}
\author{Gregory Kehne\thanks{Computer Science \& Engineering, Washington University in St. Louis, St. Louis, MO 63130. Email: {\tt kehne@wustl.edu}.}
\and Thomas Kesselheim\thanks{Institute of Computer Science and Lamarr Institute for Machine
Learning and Artificial
Intelligence, University of Bonn, Germany. Email:  {\tt thomas.kesselheim@uni-bonn.de}.}} 
\date{}

\begin{document}

	\maketitle
	
	\begin{abstract}
        Many online problems are studied in stochastic settings for which inputs are samples from a known distribution, given in advance, or from an unknown distribution. 
        Such distributions model both beyond-worst-case inputs and, when given, partial foreknowledge for the online algorithm.
        But how robust can such algorithms be to misspecification of the given distribution?
        When is this detectable, and when does it matter?
        When can algorithms give good competitive ratios both when the input distribution is as specified, and when it is not?

        We consider these questions in the setting of optimal stopping, where the cases of known and unknown distributions correspond to the well-known prophet inequality and to the secretary problem, respectively.
        Here we ask: Can a stopping rule be competitive for the i.i.d.\ prophet inequality problem and the secretary problem at the same time? 
        We constrain the Pareto frontier of simultaneous approximation ratios $(\alpha, \beta)$ that a stopping rule can attain.

        We introduce a family of algorithms that give nontrivial joint guarantees and are optimal for the extremal i.i.d.\ prophet and secretary problems.
        We also prove impossibilities, identifying $(\alpha, \beta)$ unattainable by any adaptive stopping rule.
        Our results hold for both $n$ fixed arrivals and for arrivals from a Poisson process with rate $n$. 
        We work primarily in the Poisson setting, and provide reductions between the Poisson and $n$-arrival settings that may be of broader interest.
	\end{abstract}

\section{Introduction} \label{sec:intro}

Many online problems have been studied under stochastic input models.
It is a standard assumption that in each step the input is drawn i.i.d.\ from a \emph{known} distribution. Alternatively, one can assume that the distribution is \emph{unknown}, which is similar (and in certain cases equivalent) to the input being specified by an adversary and arriving in a random order.
Examples include online matching and the Adwords problem for both the known \citep{feldman09online} and unknown \citep{devanur19near} settings, as well as online covering problems, where known vs. unknown i.i.d.\ guarantees may coincide, as they do for set cover \citep{grandoni2008set,gupta2022random}, or diverge, as for Steiner tree \citep{garg08stochastic}.
Unfortunately, algorithms designed for the known i.i.d.\ setting are often reliant on this input, and do not perform well when the distribution is unknown.
This invites the question of robustness: Can we design algorithms that exploit knowledge of the input distribution but are still competitive if the input is misspecified?

In this paper, we address this question for the i.i.d.\ prophet inequality problem.
This problem is central in the theory of optimal stopping, and has received much attention due in part to its implications for the design of online posted-price mechanisms.
This motivates both the setting where underlying distribution is known and where it is unknown.
The known-distribution setting was resolved only recently; the competitive ratio was brought above $1-\frac{1}{e}$ first to $0.738$ by \citet{abolhassani17beating}, then to $\beta_0 \approx 0.745$ by \citet{correa17posted}.\footnote{We refer to this optimal approximation ratio $\beta_0$ as the \emph{Kertz constant} \citep{kertz1986stop}. It is the unique solution to the equation $\int_0^1 \left( (\beta_0^{-1} - 1) - y (\ln y - 1)\right)^{-1} dy = 1$.}
On the other hand when the distribution is unknown, this setting was shown to coincide with the secretary problem \citep{correa19prophet}, itself a canonical problem in decision-making under uncertainty.

In this paradigm i.i.d.\ samples from a value distribution $\cD$ arrive,
and for each sample a stopping rule must decide whether to halt and receive its value or reject it forever. 
When $\cD$ is unknown, the optimal policy adopts a `wait-and-see' approach, and in the worst case receives a $\frac{1}{e}$ proportion of the expected maximum value as its (expected) reward.
On the other hand if the distribution $\cD$ is known, then the stopping rule can do better.
By maintaining a carefully chosen and decreasing acceptance threshold, it can guarantee a $\beta_0$ proportion of the expected maximum value as its reward \citep{correa17posted}.

But what can a stopping rule facing the i.i.d.\ prophet inequality problem do if its distributional foreknowledge is unreliable?
If its advice distribution $\cD'$ may not be the true value distribution $\cD$, then it can disregard $\cD'$ entirely, implement the secretary-optimal `wait-and-see' policy, and be assured of a competitive ratio of $\frac{1}{e}$. 
However this is unsatisfying, since it may only guarantee $\frac{1}{e}$ even in the event that $\cD'=\cD$ was correct after all.
In particular, the secretary-optimal policy dictates passing on the first $\frac{1}{e}$ arrivals unconditionally, even if they are excellent choices according to $\cD'$, and even if $\cD'=\cD$ is correct.
Furthermore, attempting to determine whether $\cD' = \cD$ or $\cD' \neq \cD$ itself may require passing on arrivals which the i.i.d.-prophet-optimal policy would happily choose.
Can a stopping rule strictly beat $\frac{1}{e}$ when its advice $\cD$ is correct, without forgoing all reward in the event that it is not?

To investigate this, we adopt the language of consistency and robustness  \citep{lykouris18competitive}.
For an online problem and a form of unreliable `advice' about upcoming arrivals, for which combinations of $(\alpha, \beta)$ can an algorithm be both
\begin{itemize}
    \item ($\alpha$-Consistent) $\alpha$-competitive with $\OPT$ when the advice is accurate, and
    \item ($\beta$-Robust) $\beta$-competitive with $\OPT$ when the advice is inaccurate?
\end{itemize}
While the dichotomy of $\cD' = \cD$ or $\cD' \neq \cD$ may seem blunt, the thought experiment above illustrates that the instances for which $\cD'$ and $\cD$ differ only slightly may be among the hardest.
We defer a proof to this effect and more discussion of the sensitivity of i.i.d.\ prophet stopping rules to distribution misspecification to \Cref{sec:smoothness}.
Furthermore, we argue that this dichotomy is a practically useful way of classifying inputs; for instance, what if inputs are generated via an un-auditable process which seems to perform well, but has failure modes that are not well-characterized?
Algorithmic inputs with this providence are increasingly prevalent.

We therefore set out to answer the following central questions:
\begin{quote}
    \emph{What joint consistency-robustness guarantees are possible for the i.i.d.\ prophet problem with unreliable distributional knowledge? What joint guarantees are impossible?}
\end{quote}

Note that the existing bounds can be understood as $\beta_0 \approx 0.745$-consistent and $0$-robust for the known i.i.d.\ setting, and both $\frac{1}{e}$-consistent and $\frac{1}{e}$-robust for the unknown i.i.d.\ setting. As we can randomize between the two, we immediately get the diagonal blue line in \Cref{fig:frontiers-unpop}. We also know that no algorithm is better than $\beta_0$-consistent or better than $\frac{1}{e}$-robust. We therefore want to understand what trade-offs between the parameters are possible, and what, if any, are necessary.

One should remark that---as an interesting contrast to most literature in the consistency-robustness framework (see \Cref{sec:relwork})---neither of the two extreme cases of our problem is trivial: even knowing the distribution, the correct decisions to obtain a competitive ratio of $\beta_0 \approx 0.745$ are complex, and so too are those necessary to obtain $\frac{1}{e}$ without knowing the distribution.

%%%
\subsection{Our Results}

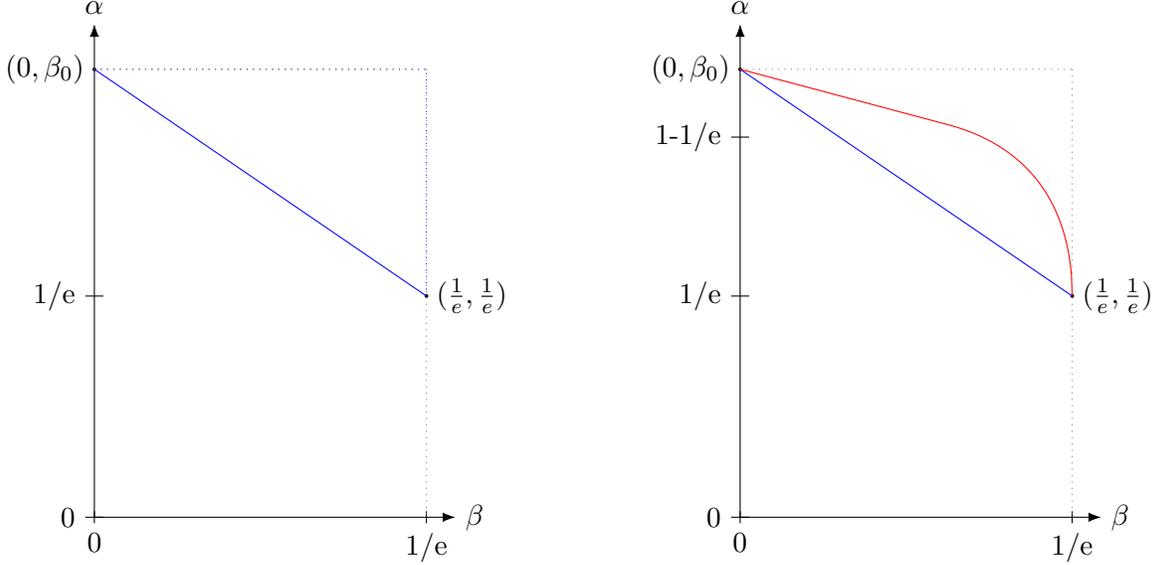
\begin{figure}
\begin{subfigure}{0.48\textwidth}
  \centering
  \begin{tikzpicture}[x=12cm, y=8cm]
    % Draw axes
    \draw[-{Latex}] (0,0) -- (0.4,0) node[right] {$\beta$};
    \draw[-{Latex}] (0,0) -- (0,0.82) node[above] {$\alpha$};
    
    % Draw ticks on x-axis
    \foreach \x in {0,1/e}
        \draw (\x,0.01) -- (\x,-0.01) node[below] {\x};
    
    % Draw ticks on y-axis
    \foreach \y in {0,1/e}
        \draw (0.01,\y) -- (-0.01,\y) node[left] {\y};
        
    % Draw points and labels
    \draw[fill=black] (0,0.745) circle[radius=0.5pt] node[left, black] {$(0,\beta_0)$};
    \draw[fill=black] ({1/exp(1)},{1/exp(1)}) circle[radius=0.5pt] node[right, black] {$(\frac{1}{e},\frac{1}{e})$};

    % Draw blue line
    \draw[blue] (0,0.745) -- ({1/exp(1)},{1/exp(1)});
    % Draw blue limit lines
    \draw[blue,dotted] (0,0.745) -- ({1/exp(1)},0.745);
    \draw[blue,dotted] ({1/exp(1)},{1/exp(1)}) -- ({1/exp(1)},0.745);

    % Draw bounding box
    \draw[gray,dotted] (0,0.745) -- ({1/exp(1)},0.745);
    \draw[gray,dotted] ({1/exp(1)},0.745) -- ({1/exp(1)},0);
\end{tikzpicture}
  \caption{The unpopulated Pareto landscape of $(\alpha,\beta)$ consistency-robustness guarantees.}
  \label{fig:frontiers-unpop}
\end{subfigure}%
\hfill
\begin{subfigure}{0.48\textwidth}
  \centering
  \begin{tikzpicture}[x=12cm, y=8cm]
    % Draw axes
    \draw[-{Latex}] (0,0) -- (0.4,0) node[right] {$\beta$};
    \draw[-{Latex}] (0,0) -- (0,0.82) node[above] {$\alpha$};
    
    % Draw ticks on x-axis
    \foreach \x in {0,1/e}
        \draw (\x,0.01) -- (\x,-0.01) node[below] {\x};
    
    % Draw ticks on y-axis
    \foreach \y in {0, 1/e, 1-1/e}
        \draw (0.01,\y) -- (-0.01,\y) node[left] {\y};
        
    % Draw points and labels
    \draw[fill=black] (0,0.745) circle[radius=0.5pt] node[left, black] {$(0,\beta_0)$};
    \draw[fill=black] ({1/exp(1)},{1/exp(1)}) circle[radius=0.5pt] node[right, black] {$(\frac{1}{e},\frac{1}{e})$};

    % Draw blue line
    \draw[blue] (0,0.745) -- ({1/exp(1)},{1/exp(1)});

    % draw red curve
    \draw[red, smooth, tension=0.5] plot[smooth] file {tikz_lambert-pts-suffix.txt};

    % draw solid red line
    \draw[red] (0,0.745) -- (0.230501,0.652724);

    % Draw bounding box
    \draw[gray,dotted] (0,0.745) -- ({1/exp(1)},0.745);
    \draw[gray,dotted] ({1/exp(1)},0.745) -- ({1/exp(1)},0);
    
\end{tikzpicture}
  \caption{The joint $(\alpha,\beta)$ consistency-robustness guarantees we obtain via \threephase{} (\Cref{thm:alg-prophet-guarantee}).}
  \label{fig:frontiers-pop-final}
\end{subfigure}
\caption{The unpopulated Pareto landscape from which we start, and our algorithmic guarantees that dominate the interpolation of the two extremes. Here $\beta_0 \approx 0.745$ is the Kertz constant.}
\label{fig:combined}
\end{figure}

On the positive side, we present a family of algorithms that provide joint guarantees.
\begin{theorem}[informal statement of \Cref{thm:alg-prophet-guarantee}]\label{thm:informal-positive-result}
    For any $\gamma \in [0, \frac{1}{e}]$, an $(\alpha,\beta)$-consistency-robustness guarantee is possible for
    \begin{align*}
        \alpha &= \gamma + e^{-\tau_1(\gamma)} - e^{-\tau_2(\gamma)} \qquad \text{and} \qquad \beta = \gamma. 
    \end{align*}
\end{theorem}
\noindent 
Here $\tau_1$ and $\tau_2$ are time thresholds given by $\tau_1(\gamma) \defeq e^{W_{-1}(-\gamma)}$ and $\tau_2(\gamma) \defeq e^{W_{0}(-\gamma)}$, where $W_0$ and $W_{-1}$ are branches of the Lambert W function.
A more complete exposition appears in \Cref{sec:algos}.

For $\gamma = 0$, the known-distribution guarantee furnished by \Cref{thm:informal-positive-result} is $\alpha=1-\frac{1}{e}$ (the known-distribution i.i.d.\ prophet setting), and it ranges to $\frac{1}{e}$ in the unknown i.i.d./secretary setting when $\gamma =\frac{1}{e}$. Interestingly, as we increase $\gamma$ from $0$, at first $\alpha$ also increases to strictly more than $1-\frac{1}{e}$, before decreasing to $\frac{1}{e}$.
Our final Pareto frontier is derived by interpolating between this and the known-distribution optimal policy of \citet{correa17posted}, and is shown in \Cref{fig:frontiers-pop-final}.

For impossibilities, we parameterize all stopping rules by a parameter $\delta \in [0,1]$, which is a rule's conditional probability of rejecting a large arrival in the starting range $t \in [0,\eps]$. 
We show:
\begin{theorem}[informal statement of \Cref{thm:negative-result}]\label{thm:informal-negative-result}
    For any (possibly adaptive) policy, the consistency-robustness guarantees must satisfy
    \begin{align*}
		\alpha \leq \beta_0 - \frac{1}{76}\left(\eps \delta  -\frac{15}{4} \eps^3 \right)^3
        \qquad \text{and} \qquad 
        \beta \leq 1 - (1 - \epsilon)(\delta^{\frac{\delta}{1-\delta}} - \delta^{\frac{1}{1-\delta}}). \notag
    \end{align*}
\end{theorem}
\noindent
For suitable combinations of $\eps$ and $\delta$, this implies both $\alpha < \beta_0$ and $\beta < \frac{1}{e}$.

\subsection{Our Techniques}
For our positive result, we design an algorithm consisting of three phases: In the first phase, from time $0$ to time $\tau_1$, all arrivals are rejected. In the second phase, from time $\tau_1$ to $\tau_2$, we accept all arrivals that are the best in the sequence so far and meet a threshold derived from $\mathcal{D}'$. In the third phase, from time $\tau_2$ to $1$, we drop this latter requirement and accept any arrival that is the best in the sequence so far. The idea behind this choice is that regardless of $\mathcal{D}'$ we only accept arrivals that are the best in the sequence so far. This ensures robustness. At the same time, if $\mathcal{D}'$ is indeed the correct distribution, we mimic the $(1-\frac{1}{e})$-competitive algorithm for i.i.d.\ prophet inequality, which sets a fixed threshold for the entire sequence (and so also only accepts arrivals that are best so far).

In the analysis, we observe that the worst incorrect advice $\mathcal{D}'$ gives a competitive ratio of $\min\{\tau_1 \log(1/\tau_1), \tau_2 \log(1/\tau_2)\}$. This is the minimum of the performances of two different secretary algorithms: One for which only the first phase serves as the observation phase, and one for which the first and the second phases serve as the observation phase. By adjusting the phase lengths, we equalize these two bounds. Our remaining degree of freedom is then the length of the second phase. To show consistency, meaning to analyze the algorithm in the case that $\mathcal{D}'$ is indeed the correct distribution, we simulate a ``Phase 0'' before the execution of the algorithm, from time $-z$ to $0$, following the same Poisson arrival rate and draws from $\mathcal{D}'$. Then, as the distributions are the same, we can view all arrivals as appearing at uniformly random times in $[-z, 1]$ across all phases.

For our negative result, we study an instance showing that one cannot get better than $\beta_0 \approx 0.745$ in the known-distribution setting defined by \cite{liu21variable}. We modify it slightly to hold in the Poisson setting with finite rate. We then argue that on this instance, any robustness will lead to strictly worse than $\beta_0$ consistency. The intuition is that, for robustness, decisions should not be too hasty because one first needs to observe the sequence. However, this instance is so tight that not accepting any arrival in the first $\epsilon$-fraction of time will necessarily lead to a competitive ratio of $\beta_0 - \Omega(\epsilon)$.

Our argument is much more involved than this because we argue about every possible policy, including adaptive policies. We show that the competitive ratio of a policy is necessarily bounded away from $\beta_0$ if it accepts the first arrival above a certain value $b$ with significant probability. Our modified distribution then introduces arrivals above $b$ at a much higher rate. This causes any policy that has good consistency to stop early on the modified instance. Additionally, in the modified instance, extremely high values arrive at a quite small rate. So these will be lost by stopping early. Importantly, there is no way to detect this higher rate, because we only argue about the acceptance of the first such arrival.

On a technical level, arguing about suboptimal policies is complicated. Letting $r(t)$ denote the expected reward of the policy after time $t$, any policy fulfills $r'(t) \geq - \int_{r(t)}^\infty n (1 - \CDF(x)) dx$, with equality for the optimal policy. This differential equation is generally quite challenging to reason about. Instead, we invoke the inverse function theorem and obtain $(r^{-1})'(s) \leq \frac{1}{- \int_s^\infty n (1 - \CDF(x)) dx}$, thus removing the recursion. This way, for each possibly adaptive policy, we can argue about how much time it spends before the expected remaining reward reaches some particular value.

%%%
\subsection{Related Work}\label{sec:relwork}

We consider the i.i.d.\ prophet inequality paradigm both in the setting where the underlying distribution is known and unknown.
The known distribution setting was resolved only recently; the competitive ratio was brought above $1-\frac{1}{e}$ first to $0.738$ by \citet{abolhassani17beating}, then to $\beta_0 \approx 0.745$ by \citet{correa17posted}, matching hard instances of \citet{hill1982comparisons} and establishing the optimal competitive ratio.
When the distribution is unknown, \citet{correa19prophet} show that the problem is essentially equivalent to the secretary problem, and show that the optimal competitive ratio is $\frac{1}{e}$ even given $o(n)$ samples from the distribution.

Since each extremal problem is well-studied, most relevant prior work falls into one of two camps: work on the i.i.d.\ prophet inequality and relaxations of the assumption of distributional foreknowledge, and work on the secretary problem augmented by additional information or advice.

In the first camp, a long line of work continues to assume the veracity of distributional prior information, but weakens the stopping rule's access to it. 
This includes the question of how many samples are needed to approach worst-case optimality.
Remarkably, one sample per arrival suffices for the prophet inequality \citep{rubinstein20optimal}.
In the i.i.d.\ setting
$1-\frac{1}{e}$ is possible using one sample per arrival \citep[Section 4.2]{correa19prophet}, and $O(1)$ samples per arrival suffice to get arbitrarily close \citep{rubinstein20optimal}. 
In a similar vein, \citet{li2023prophet} ask how useful a single quantile query to $\cD$ is, and present a stopping rule that beats $1-\frac{1}{e}$ and bears some semblance to ours (\Cref{sec:algos}).
Separately, \citet{dutting19posted} consider the case that $\cD$ and $\cD'$ differ by $\eps$ according to various statistical distances.
They normalize all arrival values to $[0,1]$ and provide stopping rules with additive guarantees in terms of $\eps$ and $n$. 

The second camp takes a perspective of algorithms with unreliable advice or predictions. The consistency-robustness paradigm for online algorithms with advice was introduced by \citet{lykouris18competitive} in the context of online caching. It has inspired hundreds of follow-up works for all kinds of problems; see \citet{mitzenmacher22algorithms} for a survey, or the Website for Algorithms with Predictions\footnote{\url{https://algorithms-with-predictions.github.io}} for a comprehensive list. In particular, augmenting the secretary problem with advice has been studied.
One line of recent work considers advice in the form of estimates of unseen arrival values, with advice inaccuracy parameterized by $\eps$.
In the consistency-robustness framework, these works aim for competitive ratio $\alpha \rightarrow 1$ as $\eps \rightarrow 0$, while maintaining robustness $\beta = \Omega(1)$ \citep{antoniadis20secretary, fujii24secretary, choo24short, balkanski24fair}.
An exception is \citet{braun24secretary}, who use a prediction of the difference between the largest and $k^{\text{th}}$-largest values to be $(1 - \Omega(1))$-competitive.

Another line of work provides distributional information to the secretary problem.
For instance, \citet{correa21secretary} and \citet{kaplan20competitive} take advice to be a number of samples without replacement from an adversarially chosen collection of values, the remainder of which then arrive in random order, and \citet{dutting21secretaries} develop an LP-based approach to handle this and other augmentations.
These models assume that augmentations are generated as prescribed, and do not allow the possibility that this distributional side-information is adversarial or otherwise faulty.

While it is arguably more common to assume a fixed number of arrivals in optimal stopping problems, there is also quite some work assuming Poisson arrivals. In particular, 
\citet{allaart07prophet} undertakes a careful analysis of the optimal stopping problems studied in the discrete-time setting by \citet{krengel1978semiamarts} and \citet{hill1982comparisons} in the limit as the Poisson arrival rate becomes large, and \cite[Section 10.3]{singla18combinatorial} presents a concise proof of the i.i.d.\ optimal guarantee in the Poisson setting.
It is also the natural setting for recent work on steady-state analogs of these stopping problems \citep{kessel22stationary}. 
We find the Poisson setting better suited to proving impossibilities, and present our results for Poisson arrivals. 
However this is without loss of generality, and our positive and negative results also hold in the standard $n$-arrival model as $n$ becomes large.
We formalize this via reductions between the settings in \Cref{sec:prelims}. 
While the observation these settings are roughly equivalent is not novel, we are unaware of prior reductions that are comparably general.

Finally, in independent work, \citet{bai2025optimal} also jointly consider the secretary and prophet inequality problems from the perspective of consistency-robustness tradeoffs.
Like us, they consider the expected competitive ratio, which they refer to as MaxExp, as well as MaxProb, the objective of maximizing the probability of stopping on the maximum value in  the arrival sequence.
Their family of stopping rules for the expected competitive ratio is similar to ours: it accepts arrivals that are best-so-far and operates in three phases, where the first phase rejects all arrivals and the second phase applies an additional threshold that arrivals must surpass in order to be accepted. 
In our work this middle phase threshold is stochastic but constant, while for \citeauthor{bai2025optimal} it is deterministic but decreasing.
Both approaches eventually interpolate with Hill and Kertz' optimal strategy in order to attain sufficiently high consistency.
While our stopping rules are similar, our impossibility results are complementary: we show a hardness result for the expected competitive ratio, while they show hardness for the MaxProb objective.

%%%
\subsection{Organization}
The rest of the paper proceeds as follows.
In \Cref{sec:prelims} we formally introduce the Poisson and $n$-arrival settings the notation we will use, and reductions between the two settings.
In \Cref{sec:algos} we present our three-phase stopping rule \threephase{}, analyze its performance, and use it to establish nontrivial lower bounds on the Pareto frontier.
In \Cref{sec:hardness} we present consistency-robustness impossibilities for this problem, showing that there are pairs of distributions $\cD$, $\cD'$ for which \emph{no} adaptive stopping rule can simultaneously achieve worst-case optimal competitive ratios for the known- and unknown-distribution i.i.d.\ prophet inequality problems.
\Cref{sec:smoothness} briefly addresses whether stopping rules can use distributional advice $\cD' \neq \cD$ if $\cD'$ and $\cD$ are close, and uses the reductions between settings together with prior $n$-arrival impossibility results to provide a negative answer.
We conclude by discussing the generalization of our approach to related settings and future directions in \Cref{sec:conclusion}.

\section{Preliminaries} \label{sec:prelims}

The policy $\ALG$ receives some advice value distribution $\cDp$ which is possibly distinct from the true value distribution $\cD$.
We will assume that the value distributions $\cD$ and $\cDp$ are non-atomic; since we consider the class of randomized stopping rules (which may implement randomized tie-breaking), this is without loss of generality.

For a given rate/number of arrivals, a policy $\ALG$ that knows $\cD$ is \emph{$\alpha$-consistent} if 
\[
    \expectover{\bX \sim \cD}{\ALG(\bX, \cD)} \geq \alpha \cdot \expectover{\bX \sim \cD}{\MAX(\bX)}
\] 
for all value distributions $\cD$, where we use $\bX\sim \cD$ to denote a random sequence of arrivals $\bX$, each of which is an i.i.d.\ sample from $\cD$, and the expectation is also over the randomness of $\ALG$.
A policy is \emph{$\beta$-robust} if 
\[
    \expectover{\bX \sim \cD}{\ALG(\bX, \cDp)} \geq \beta \cdot \expectover{\bX \sim \cD}{\MAX(\bX)}
\] 
for all value distributions $\cD$ and $\cDp$, where again the expectation is over $\bX$ and $\ALG$.

What is the process by which $\bX$ is generated from $\cD$?
We will consider two choices.

\subsection{The \texorpdfstring{$n$}{n}-Arrival Model}
Here we consider $n$ discrete time steps $i \in [n]$ and $\bX \defeq (X_1, \ldots, X_n)$, where each $X_i$ is an i.i.d.\ sample from $\cD$.
Although distributional assumptions differ, this is the most common arrival model in which the secretary problem, the prophet inequality, and related problems are studied.

\subsection{The Poisson Arrival Model}
Here $\bX$ is comprised of arrivals that follow a Poisson point process with rate $\lambda=n$.
For each arrival a value sampled i.i.d.\ from a distribution $\cD$ at a given rate $\lambda$, as time ranges over $t \in [0,1]$.
In other words, in each infinitesimal time interval $dt$, an independent sample from $\cD$ arrives with probability $\lambda \cdot dt$.
This gives rise to a sequence of arrival values and times. 
We let $\bT$ denote the random set of arrival times according to the Poisson process, and $\bX \defeq (X_t)_{t \in \bT}$ the random sequence of values at these times.

%%%
\subsection{Reductions Between i.i.d.\ Arrival Models}
\label{sec:reductions}

We now present a pair of reductions between these arrival models.
These will enable us to state algorithmic and impossibility results for both models. 
These reductions may well be folklore; however we are unaware of prior appearances, and so present them here.
Proofs appear in \Cref{app:prelims}.

The first reduction shows that a competitive stopping rule in the Poisson setting implies a competitive stopping rule in the $n$-arrival setting. 
\begin{restatable}[$n$ to Poisson]{lemma}{ntopoisson}
	\label{lem:good-poisson-implies-good-fixed-n}
	Given $n \in \mathbb{N}$, let $\lambda = n(1-\eps)$ for $\eps = \sqrt{\frac{\log n}{n}}$. 
    Suppose $\ALG_\lambda$ is an $\alpha$-competitive Poisson stopping rule for rate $\lambda$. 
	Let $\ALG_n$ be the $n$-arrival stopping rule that simulates $\ALG_\lambda$ by first sampling Poisson arrival times and then taking the first $n$ values to be $X_1, \ldots, X_n$.
	Then $\ALG_n$ is $(\alpha - 4\eps)$-competitive in the $n$-arrival setting.
\end{restatable}

We can also make claims in the other direction: a competitive stopping rule in the $n$-arrival setting implies a competitive rule for the Poisson setting.
Here we again make no assumptions on the form of the $n$-arrival stopping rule, but the derived Poisson rule is for a much smaller rate.

\begin{restatable}[Poisson to $n$]{lemma}{poissonton}
\label{lem:good-fixed-n-implies-good-poisson}
    Suppose $\alg_n$ is stopping rule for $n$ fixed arrivals that is $\alpha$-competitive.
    Then for any $\eps \in (0,1]$, we may use it to construct a Poisson stopping rule $\alg_\lambda$ that is $\alpha\cdot(1-\eps/2)$-competitive on Poisson arrivals at rate $\lambda = \eps n$.
\end{restatable}

Note this reduction holds for both the known-distribution and unknown-distribution settings.
If $\alg_n$ is a known-distribution stopping rule, then given Poisson arrivals from $\cD$, then $\alg_\lambda$ proceeds by simulating $\alg_n$ on arrivals from $\tilde \cD$. 

For impossibilities, we will use the contrapositive of \cref{lem:good-fixed-n-implies-good-poisson}.
\begin{corollary} \label{lem:poisson-imposs-implies-fixed-n-imposs}
    If no $\alpha$-competitive stopping rule exists for the i.i.d.\ prophet inequality problem with Poisson arrivals at rate $\lambda$ from a known (unknown) distribution, then no $\frac{\alpha}{1-\eps/2}$-competitive stopping rule exists for the prophet inequality problem with $n = \lambda/\eps$ i.i.d.\ arrivals from a known (unknown) distribution.
\end{corollary}

\section{Beating the Linear Frontier} \label{sec:algos}

How can we do better than randomly interpolating between the secretary-optimal stopping rule and the known-distribution i.i.d.\ stopping rule, and improve upon the blue line in \Cref{fig:frontiers-unpop}? 
Indeed, why should this be possible?
Working in the Poisson arrivals model, we now introduce an algorithm which beats the line. 

%%%
\subsection{A Three-Phase Algorithm}

Our algorithm handles the Poisson arrivals for $t\in [0,1]$ by partitioning the arrival interval into three phases, and treating arrivals differently depending on the phase in which they arrive. 
This is preceded by a `Phase 0' of simulated arrivals from the advice distribution $\cDp$, which condition the algorithm with information about the advice. 
Our algorithm, \threephase{}, does the following:
\begin{enumerate}
    \item[P0.] $t \in [-z,\: 0)$: Simulate arrivals from $\cDp$ at rate $n$.
    \item[P1.] $t \in [0,\:\tau_1)$: Do not accept.
    \item[P2.] $t \in [\tau_1,\:\tau_2)$: Accept if greater than the $\ell$-th largest arrival in P0 and all arrivals in P1. 
    \item[P3.] $t \in [\tau_2,\:1]$: Accept if greater than all arrivals in P1 and P2.
\end{enumerate}
The parameters $z$, $\ell(z)$, $\tau_1$, and $\tau_2$ will be specified in \Cref{sec:alg-approach}. 
We will set $\ell(z)$, $\tau_1$, and $\tau_2$ in pursuit of provable guarantees, and then take the limit as $z$ becomes large.

%%%
\subsubsection{Notation}
Throughout this section, we will make use of specific notation. We split the arrivals from the Poisson process with rate $n$ and distribution $\cD$ into $\bX_1$, $\bX_2$, and $\bX_3$ depending on the phase during which they arrive. 
In this section, we will consider a (simulated) Poisson process with rate $n$ and distribution $\cDp$ generating arrivals $\bX_0$ over the time period $t \in [-z, 0]$. 

For each group of arrivals $\bX_i$, we let $x_i^{(j)}$ denote the \emph{top} order statistics in $\bX_i$, so that $x_i^{(j)}$ is the $j^{th}$ \emph{largest} value in $\bX_i$, or $0$ if $|\bX_j| < j$. 
Then $x_i^{(1)}$ is the largest arrival in $\bX_i$. 
Furthermore, for $\bX = \bX_0 \cup \bX_1 \cup \bX_2 \cup \bX_3$ the arrivals across all phases, let $x^{(j)}$ denote the $j^{th}$ largest arrival in $\bX$.
For arrivals with $t \geq 0$, let $\bX_+ \defeq \bX_1 \cup \bX_2 \cup \bX_3$, and let $x^{(j)}_+$ be the $j^{th}$ largest arrival in $\bX_+$.

%%%
\subsubsection{Intuition} \label{sec:three-phase-intuition}

To understand why these phases improve upon the Pareto frontier, it is helpful to informally take a Goldilocks perspective and consider three cases: when the advice $\cDp$ is too low relative to the true distribution $\cD$, when it is too high, and when it is accurate---that is, when it is just right.
\begin{itemize}
    \item When $\cDp$ is `too low', the $\ell$-th order statistic $x_0^{(\ell)}$ is too small to serve as a meaningful threshold for selection in Phase 2. In this case Phase 2 and Phase 3 are the same, and we are running a secretary strategy with threshold $\tau_1$; that is, Phase 1 is the observation phase, and Phases 2 and 3 are the selection phase.
    \item If $\cDp$ is `too high', then $x_0^{(\ell)}$ is too large, and the algorithm makes no selections during Phase 2. In this case, we are again running a secretary strategy, where Phases 1 and 2 are the observation phase and Phase 3 is the selection phase.
    \item If $\cDp$ is `just right', then $\cDp = \cD$ or at least the relevant top quantiles match up. In this case, $x_0^{(\ell)}$ represents a good static threshold. If the maximum $x_1^{(1)}$ of Phase 1 is smaller than it, then we implement our good threshold for Phase 2 and relax to the lower `best so far' threshold for Phase 3. On the other hand if $x_1^{(1)} > x_0^{(\ell)}$ we use it as our threshold for Phase 2 and Phase 3.
\end{itemize}
The first two cases provide an intuitive grounding for our secretary guarantee in \Cref{sec:three-phase-secretary}, while the third explains why we see improved performance in the prophet setting in \Cref{sec:three-phase-prophet}.

\subsubsection{Special Cases} \label{sec:special-cases}
It bears mentioning that we recover the secretary-optimal stopping rule and a sample-based variant of the $(1-\frac{1}{e})$-competitive fixed-threshold algorithm of \cite{ehsani18prophet}, depending on our choices of $\tau_1$ and $\tau_2$. 
If $\tau_1 = \tau_2 = \frac{1}{e}$ then we fully disregard the samples from $\cDp$ and recover the secretary algorithm; on the other hand if $\tau_1 = 0$ and $\tau_2 = 1$ then we recover the form and the guarantee of the static quantile-threshold prophet algorithm.

But notably, while this choice of $\tau_1 =0$, $\tau_2 = 1$, and $\ell(z) \approx z$ recovers the competitive ratio of $1-\frac{1}{e}$, it is \emph{not} the prophet-optimal algorithm in this family. Setting $\tau_1=0$, $\tau_2<1$, and $\ell(z)\approx z$ gives a strictly better guarantee, as do a range of choices of $\tau_1>0$ and $\tau_2<1$. 
%This is illustrated by the positive slope of the prophet guarantee for \threephase{} as a function of $\gamma$ in \Cref{fig:frontiers-populated}. 
While this is perhaps surprising, it is consistent with the work of \citet{li2023prophet}, who show one can strictly outperform $1-\frac{1}{e}$ in the i.i.d.\ prophet setting using only a single quantile query to $\cD$ via a stopping rule that deploys a fixed threshold followed by a `best so far' phase. This corresponds to an algorithm in our family with parameters $\tau_1 = 0$ and $\tau_2 < 1$.

In order to simultaneously provide a secretary guarantee, we will relate our choices of $\tau_1$ and $\tau_2$ and choose $\tau_1 > 0$.
Nevertheless \threephase{} still outperforms $1-\frac{1}{e}$ in a manner similar to the work of \citeauthor{li2023prophet} for a large range of values of $\tau_1$ and $\tau_2$, as illustrated in \cref{fig:frontiers-populated}.

%%%
\subsection{Approach and Parameters} \label{sec:alg-approach}
We will relate our choice of phase boundaries $\tau_1$ and $\tau_2$ according to
\begin{equation} \label{eq:ass-tau1-tau2}
    \tau_1 \log \left(\frac{1}{\tau_1}\right) = \tau_2 \log \left(\frac{1}{\tau_2}\right) = \gamma
\end{equation}
for $\gamma \in [0,\frac{1}{e}]$. 
So for a chosen $\gamma$, $\tau_1$ and $\tau_2$ are given by 
\[
    \tau_1(\gamma) \defeq e^{W_{-1}(-\gamma)}, \qquad \qquad \tau_2(\gamma) \defeq e^{W_{0}(-\gamma)},
\]
where $W_0$ and $W_{-1}$ are branches of Lambert's W function, which give inverses to the function $y(x) = x e^x$. We repurpose this to give inverses to $y(x) = x \log x$ for $x \in [0,1]$, as in \cref{fig:taus}.

\begin{figure}[t]
    \centering
    \begin{tikzpicture}[scale=8]
        % Draw axes
        \draw[->] (0,0) -- (1.2,0) node[right] {$t$};
        \draw[->] (0,0) -- (0,0.5) node[above] {$\gamma$};
        
        % Add tick marks and labels with smaller ticks
        \draw (0,0.02) -- (0,-0.02) node[below] {$0$};
        % \draw (0.25,0.02) -- (0.25,-0.02) node[below] {$\frac{1}{4}$};
        % \draw (0.5,0.02) -- (0.5,-0.02) node[below] {$\frac{1}{2}$};
        % \draw (0.75,0.02) -- (0.75,-0.02) node[below] {$\frac{3}{4}$};
        \draw (1,0.02) -- (1,-0.02) node[below] {$1$};
        
        \draw (0.02,0) -- (-0.02,0) node[left] {$0$};
        \draw (0.02,0.25) -- (-0.02,0.25) node[left] {$\gamma(\tau_1) = \gamma(\tau_2)$};
        % \draw (0.02,0.5) -- (-0.02,0.5) node[left] {$\frac{1}{2}$};
        
        % Plot the function y = x ln(1/x)
        \draw[smooth, thick, blue, domain=0.001:1, samples=200] 
            plot (\x,{\x*ln(1/\x)}) 
            node[above, yshift=1cm] {$\gamma=\tau\log(1/\tau)$};
        
        % Mark the point at x = 1/e
        \def\e{2.718281828459045}
        \coordinate (P) at ({1/\e},{(1/\e)*ln(\e)});
        \filldraw (P) circle (0.2pt);
        % Add tick and label for 1/e on both axes
        \draw ({1/\e},0.02) -- ({1/\e},-0.02) node[below] {$\frac{1}{e}$};
        \draw (0.02,{1/\e}) -- (-0.02,{1/\e}) node[left] {$\frac{1}{e}$};
        
        % Add two symmetric points at y = 0.3 (approximately)
        % We need to solve x*ln(1/x) = 0.3 numerically
        % The solutions are approximately x ≈ 0.157 and x ≈ 0.843
        \coordinate (P1) at (0.116101,0.25);
        \coordinate (P2) at (0.699491,0.25);
        \filldraw (P1) circle (0.2pt);
        \filldraw (P2) circle (0.2pt);
        % Add tick and label for \tau values
        \draw (0.116101,0.02) -- (0.116101,-0.02) node[below] {$\tau_1$};
        \draw (0.699491,0.02) -- (0.699491,-0.02) node[below] {$\tau_2$};
        
        % Add dotted lines connecting the points
        \draw[dashed, gray] (0.116101,0) -- (0.116101,0.25) -- (0.699491,0.25) -- (0.699491,0);
    \end{tikzpicture}

    \caption{The relationship between choices of $\tau_1$ and $\tau_2$ and the secretary competitive ratio guarantee $\gamma$ given by \cref{lem:three-phase-prob-ratio-bound}. Observe that if $\tau_1 = \tau_2 = 1/e$ then we recover the secretary optimal stopping rule and the secretary guarantee.}
    \label{fig:taus}
\end{figure}
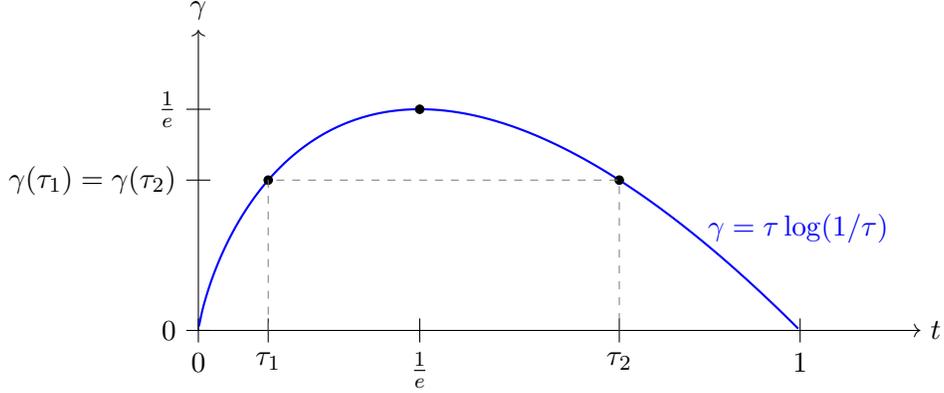

For this choice we will prove our joint guarantees. 
For the secretary setting, we will show

\begin{restatable}{theorem}{algSecretaryGuarantee}
% \begin{theorem}
\label{thm:alg-secretary-guarantee}
    If $\cD \neq \cD'$ then \threephase{} with $\tau_1(\gamma)$ and $\tau_2(\gamma)$ set according to \cref{eq:ass-tau1-tau2} satisfies
    \[
        \expect{\ALG(\bX)} \geq \gamma \cdot \expect{\MAX(\bX_+)}.
    \]
% \end{theorem}
\end{restatable}
\noindent

For the prophet setting we will show
\begin{restatable}{theorem}{threephaseProphetGuarantee}
% \begin{theorem} 
\label{thm:alg-prophet-guarantee}
    If $\cD = \cD'$ then \threephase{} with $\tau_1(\gamma)$ and $\tau_2(\gamma)$ set according to \cref{eq:ass-tau1-tau2} with $\gamma \in [0, \frac{1}{e})$ satisfies
    \[
        \expect{\ALG(\bX_0, \bX_+)} \geq \left( \gamma + e^{-\tau_1(\gamma)} - e^{-\tau_2(\gamma)} - O\left(z^{-1}\right) \right) \cdot \expect{\MAX(\bX_+)}
    \]
    for all sufficiently large $z$.
% \end{theorem}
\end{restatable}

The remainder of the section is dedicated to proving these theorems.

To prove that $\expect{\ALG(\bX_+)} \geq \alpha \cdot \expect{\ALG(\bX_+)}$ for a sequence of arrivals $\bX_+$, it is often helpful to prove a stronger \emph{approximate stochastic dominance} property:
\begin{equation}
    \Pr{\ALG(\bX_+) \geq y} \geq \alpha \cdot \Pr{\MAX(\bX_+) \geq y} \qquad \qquad \text{ for all } y \in \bbR_{\geq 0}. \label{eq:approx-stoch-dom}
\end{equation}
This can be shown by reasoning about $\ALG$ in terms of the quantiles of the distribution of arrivals; integrating both sides over the distribution gives a competitive ratio guarantee. 
This is an established approach for proving competitive ratios in both the i.i.d.\ prophet and prophet secretary settings \cite{singla18combinatorial,correa19prophet}. 

For the prophet setting, in \Cref{sec:three-phase-prophet} we will actually provide an `ordinal' approximate stochastic dominance guarantee of the form  
\begin{equation}
    \Pr{\ALG(\bX_0, \bX_+) \geq x^{(i)}} \geq \alpha \cdot \Pr{\MAX(\bX_+) \geq x^{(i)}} \qquad \forall i \in \bbN,  \label{eq:approx-stoch-dom-ordered}
\end{equation}
where the order statistics $x^{(i)}$ are over $\bX = \bX_0 \cup \bX_+$, and we say that $x^{(i)} = 0$ if $i > |\bX|$.

Note that \eqref{eq:approx-stoch-dom-ordered} need not imply a corresponding competitive ratio guarantee, since the event that $\ALG$ exceeds $x^{(i)}$ could positively correlate with lower realized values of the random variable $x^{(i)}$. 
However because $\ALG$ only makes (pairwise) ordinal comparisons within a single Poisson process sequence $\bX$ in the prophet setting, the likelihood that it accepts any given $x^{(i)}$ is independent of the value that $x^{(i)}$ takes.
In \Cref{sec:alg-competitive-ratios} we show that in the prophet setting, \eqref{eq:approx-stoch-dom-ordered} implies scalar stochastic dominance of the form \eqref{eq:approx-stoch-dom}, and therefore gives a competitive ratio guarantee.

This problem is more acute for our analysis in the secretary setting.
For the standard secretary-optimal stopping rule in the unknown i.i.d.\ setting, a guarantee of the form 
\begin{equation}
    \Pr{\ALG(\bX_0, \bX_+) = \MAX(\bX_+)} \geq \gamma   \label{eq:secretary-prob-guarantee}    
\end{equation}
translates directly to a competitive ratio lower bound, for the same reason: because it decides whether to stop based only on ordinal comparisons between arrivals in $\bX_+$, its probability of stopping on $x_+^{(1)}$ is independent of the value that $x_+^{(1)}$ takes.
However \threephase{} uses a threshold derived from $\bX_0$, and does not decide whether to stop based only on ordinal comparisons within $\bX_+$.
We circumvent this by proving in \Cref{sec:three-phase-secretary} that \eqref{eq:secretary-prob-guarantee} holds even conditioned on the maximum value.

In the following proofs, it will be convenient to assume that \emph{all} order statistics are defined (letting $x^{(i)} = 0$ for all $i > |\bX|$) and that $x^{(i)} > x^{(j)}$ for all such $i < j$ for the purposes of the algorithm. 
Furthermore, it will be useful to assume that each such `zero arrival' is at a uniformly random time. 
This allows us to suppose that each arrival in $[0,1]$ has a largest prior arrival $x^{(k)}$, and that the threshold $x_0^{(\ell)}$ is always well-defined.
Therefore the analysis that follows is technically for a variant of \threephase{} that has some probability of halting early in Phase 2 with reward 0 if $\bX_0 = \emptyset$ and $\bX_1 = \emptyset$, or halting early in Phase 3 with reward 0 if $\bX_1 = \emptyset$ and $\bX_2 = \emptyset$, and otherwise behaves identically. 
As such, the guarantees will be bona fide lower bounds on its performance.

Our analysis will also make use of the fact that \threephase{} has the \emph{best so far} property (or `regret-free' \cite{buchbinder10secretary}), which it shares with both the fixed-threshold and the secretary stopping rules (but notably \emph{not} the optimal i.i.d.\ stopping rule):
\begin{observation} \label{obs:best-so-far}
    If \threephase{} stops on an arrival $x_t$, then it is the largest arrival in $[0,t]$.
\end{observation}

For the case when $\cDp = \cD$, our analysis proceeds by analyzing the probability that $\ALG$ accepts each of $x^{(1)}, x^{(2)}, \ldots, x^{(\ell)}$, conditioned on it arriving at time $t \in [0,1]$. That is, for each $j \in \{1, \ldots, \ell\}$ and $t \in [\tau_1,1]$, we will proceed by getting a handle on the conditional probability 
\[
    \Pr{\ALG(\bX) = x^{(j)}\given x_t = x^{(j)}},
\]
and use this to provide a bound along the lines of \eqref{eq:approx-stoch-dom-ordered}.

For our proofs we will further by restrict the parameters of \threephase{}. 
It is possible performance guarantees could be improved by changing either of these assumptions.
\begin{assumption} \label{ass:a-value}
    The Phase 0 threshold order statistic is given by $\ell \defeq z + 2$; that is, for $\ell - 1 = a(z+1)$ we have that $a = 1$.
\end{assumption}

Throughout the remainder of section, omitted proofs appear in \Cref{app:algos}.

%%%
\subsection{\threephase{} in the Secretary Setting (\texorpdfstring{$\cDp \neq \cD$}{D' != D})} \label{sec:three-phase-secretary}

We first consider the case when $\cD \neq \cDp$. 
When the advice distribution is inaccurate, we will recover a guarantee for \threephase{} that corresponds to the secretary policy, using $t = \tau_1$ or $t = \tau_2$ as the cutoff between the observation and selection phases.

\begin{lemma} \label{lem:three-phase-prob-ratio-bound}
    For any $\cD$ and $\cD'$, and for any $y \in \bbR_{\geq 0}$, \threephase{} guarantees that 
    \[
         \Pr{\ALG(\bX_0, \bX_+) = \MAX(\bX_+) \given \MAX(\bX_+) = y} \geq \gamma(\tau_1, \tau_2), 
    \]
    where $\bX_0$ is Poisson arrivals from $\cDp$, $\bX_+$ is  Poisson arrivals from $\cD$, and
    \[
        \gamma(\tau_1, \tau_2) \defeq \min\left(\tau_1 \log \left( \frac{1}{\tau_1} \right), \: \tau_2 \log \left( \frac{1}{\tau_2} \right)\right).
    \]
\end{lemma}
In fact we will show something stronger: that for all $\cD$ and \emph{any} threshold $L\in \bbR_{\geq 0}$, if $L = x_0^{(\ell)}$ is the threshold chosen by \threephase{} in Phase 0, then
\[
     \Pr{\ALG(L, \bX_+) = \MAX(\bX_+) \given \MAX(\bX_+) = y} \geq \gamma(\tau_1, \tau_2). 
\]
\begin{proof}
    We will write $x_+^{(1)} = \MAX(\bX_+)$ and $L = x_0^{(\ell)}$ for convenience. 
    For fixed $t \in [0,1]$, let $x^{(k)}$ denote the largest arrival in $[0,t)$ in $\bX_+$.
    We have two cases, corresponding to whether $y \leq L$.
    
    If $y \leq L$ then since the arrival time of $x_+^{(1)}$ is uniformly distributed in $[0,1]$, 
    \begin{align}
        \Pr{\ALG(L,\bX_+)= x^{(1)}_+ \given x_+^{(1)} = y } 
        &= \int_0^1 \Pr{\ALG(\bX)= x^{(1)}_+ \given x_t = x_+^{(1)} \wedge x_+^{(1)} = y} \:dt \notag \\
        &= \int_{\tau_2}^1 \Pr{x^{(k)} \text{ arrives in } [0,\tau_2] \given x_t = x_+^{(1)} \wedge x_+^{(1)} = y} \:dt \notag \\
        &= \int_{\tau_2}^1 \frac{\tau_2}{t} \:dt = \tau_2 \log\left(\frac{1}{\tau_2}\right), \label{eq:secbound-high-thresh}
    \end{align}
    since for \threephase{} the arrival $x_t = x_+^{(1)} \leq L$ is selected if and only if $t \geq \tau_2$ and the largest arrival prior to $x_t = x_+^{(1)}$ falls in $[0,\tau_2]$. 

    On the other hand, if $y > L$ then since the arrival of $x_+^{(1)}$ is uniformly distributed in $[0,1]$, 
    \begin{align}
        &\Pr{\ALG(L,\bX_+)= x_+^{(1)} \given x_+^{(1)} = y }
        = \int_0^1 \Pr{\ALG(\bX)= x^{(1)}_+ \given x_t = x_+^{(1)} \wedge x_+^{(1)} = y} \:dt \notag \\
        &\qquad\qquad= \int_{\tau_1}^1 1 - \Pr{x^{(k)} \text{ arr. in } [\tau_1,\min(\tau_2,t)) \wedge x^{(k)} > L \given  x_t = x_+^{(1)} \wedge x_+^{(1)} = y } \:dt \notag\\
        &\qquad\qquad\qquad\qquad - \int_{\tau_2}^{1} \Pr{x^{(k)} \text{ arr. in } [\tau_2, t] \given x_t = x_+^{(1)} \wedge x_+^{(1)} = y } \:dt . \notag
        \intertext{This is because $\ALG$ fails to accept $x_+^{(1)}$ upon arrival in $[\tau_1, 1]$ if and only if either $x^{(k)}$ arrives in $[\tau_1, \tau_2]$ \emph{and} beats the threshold $L$, or if $x^{(k)}$ arrives in $[\tau_2, t]$. This probability is only smaller if we remove the condition that $x^{(k)}$ must beat the threshold in order to be accepted in Phase 2:}
        &\qquad\qquad\geq \int_{\tau_1}^1 1 - \Pr{x^{(k)} \text{ arr. in } [\tau_1, t) \given x_t = x_+^{(1)} \wedge x_+^{(1)} = y } \:dt \notag \\
        &\qquad\qquad= \int_{\tau_1}^1 \frac{\tau_1}{t} \:dt = \tau_1 \log\left(\frac{1}{\tau_1}\right). \label{eq:secbound-low-thresh}
    \end{align}
    
    Taken together, \eqref{eq:secbound-high-thresh} and \eqref{eq:secbound-low-thresh} imply that for any $y \in \bbR_{\geq 0}$, 
    \begin{align}
        \Pr{\ALG(L,\bX_+)= \MAX(\bX_+) \given \MAX(\bX_+) = y}
        &\geq \min \left( \tau_1 \log \left(\frac{1}{\tau_1}\right) , \: \tau_2 \log \left(\frac{1}{\tau_2}\right) \right) = \gamma(\tau_1,\tau_2). \label{eq:sec-condnl-prob-guarantee} \qedhere
    \end{align}
\end{proof}

We next analyze the performance of the algorithm in the prophet setting.

%%%
\subsection{\threephase{} in the Prophet Setting (\texorpdfstring{$\cDp = \cD$}{D' = D})} \label{sec:three-phase-prophet}

We will do case analysis based on whether the value that $\ALG$ stops on is above or below the Phase 0 threshold $x_0^{(\ell)}$.
We first bound the probability \threephase{} stops on $x^{(j)}$ for $j \leq \ell$
\begin{restatable}{lemma}{conditionalprobforlowXj}
% \begin{lemma}
\label{lem:condnl-prob-low-j}
For all $j \leq \ell$ and all $t \geq \tau_1$, we have
\[
    \Pr{\ALG(\bX) = x^{(j)}\given x_t = x^{(j)}} \geq \left(1- \frac{t}{z + 1}\right)^{j - 1} \cdot \frac{\tau_1}{t}  + \left( 1- \frac{t}{z + 1}\right)^{\ell - 1} \cdot \left( \frac{\min(\tau_2, t)}{t} - \frac{\tau_1}{t}\right).
\]
% \end{lemma}
\end{restatable}

The intuition for this bound is as follows: suppose we at time $t$ and the $j$-th best value of the entire sequence arrives. 
What needs to happen for \threephase{} to accept it? 
None of the top $j-1$ best elements may arrive between $0$ and $t$. 
Additionally, we may not have accepted anything else so far because it was not the best so far or because it did not meet the threshold.

We establish a similar bound on the probability that \threephase{} stops on $x^{(j)}$ for $j > \ell$.
\begin{restatable}{lemma}{conditionalprobforhighXj}
% \begin{lemma} 
\label{lem:condnl-prob-high-j}
    For $j > \ell$ and $t \geq \tau_2$, we have
    \begin{align*}
        \Pr{\ALG(\bX) = x^{(j)}\given x_t = x^{(j)}} \geq \left(1 - \frac{t}{z+1}\right)^{j-1} \frac{\tau_2}{t}.
    \end{align*}
% \end{lemma}
\end{restatable}
This reflects that if we have reached Phase 3, then the $j$-th largest value is accepted provided that it is the largest so far, and also that the largest arrival in $[0,t)$ prior to it arrives outside of $[\tau_2, t)$.

%%%
\subsubsection{Ordinal Approximate Stochastic Dominance}
With these conditional probabilities in hand, we are ready to make claims about the relative probabilities that $\ALG$ and $\MAX$ are within the top order statistics $x^{(1)}, \ldots, x^{(\ell')}$ for various $\ell'$. 

\begin{restatable}{lemma}{prophetProbabilityRatioTop}
% \begin{lemma} 
\label{lem:prophet-prob-ratio-top}
    For any fixed $\tau_1 \neq \tau_2$ and sufficiently large $z$, for all $\ell' \leq \ell$, the probability that $\ALG$ chooses an arrival among the top $\ell'$ is at least 
    \[
        \Pr{\ALG(\bX) \geq x^{(\ell')}} \geq \alpha(\tau_1, \tau_2, z) \cdot \Pr{\MAX(\bX_+) \geq x^{(\ell')}},
    \]
    where 
    \[
        \alpha(\tau_1, \tau_2, z) \defeq \tau_1 \log\left( \frac{1}{\tau_1}\right) + e^{-\tau_1} - e^{-\tau_2} -O\left(z^{-1}\right).
    \]
% \end{lemma}
\end{restatable}
We derive this by integrating the probability guarantees from \Cref{lem:condnl-prob-low-j} over time in order to obtain unconditional probabilities. 
The claim follows from computing and carefully bounding these integrals, then identifying two distinct upper bounds on the cumulative maximum probability term.

We will also need to establish a lower bound for $\ell' > \ell$.
For the constituent terms $x^{(j)}$ for $j \leq \ell$, we may use \eqref{eq:prophet-prob-ratio} and \cref{lem:prophet-prob-ratio-top}. 
But for $j > \ell$ we will use the bound from \Cref{lem:condnl-prob-low-j} in order to get a handle on the $\Pr{\ALG(\bX) = x^{(j)}}$ terms.

\begin{restatable}{lemma}{prophetProbRatioBottom}
% \begin{lemma} 
\label{lem:prophet-prob-ratio-bottom}
    For any fixed $\tau_1 \neq \tau_2$ and sufficiently large $z$, for all $\ell' > \ell$, the probability that $\ALG$ chooses an arrival among the top $\ell'$ is again at least 
    \[
        \Pr{\ALG(\bX) \geq x^{(\ell')}} \geq \alpha(\tau_1, \tau_2, z) \cdot \Pr{\MAX(\bX_+) \geq x^{(\ell')}}.
    \]
% \end{lemma}
\end{restatable}

This proof proceeds starting from \eqref{eq:low-l'-alg-bound} with $\ell' = \ell$, which satisfies the condition $\ell'\leq \ell$. 
We will apply \Cref{lem:condnl-prob-high-j} and again integrate over $t \in [\tau_2, 1]$ to get probability bounds for $j > \ell$. 
Taking $a = 1$ in accordance with \Cref{ass:a-value} and collecting terms yields the claim.
We provide a proof in \Cref{app:algos}.

Since this holds for both $\ell' \leq \ell$ and $\ell' > \ell$, \cref{lem:prophet-prob-ratio-top,lem:prophet-prob-ratio-bottom} together imply
\begin{corollary} \label{lem:prophet-prob-ratio}
    If $\cD = \cD'$ then for any fixed $\tau_1 \neq \tau_2$ and sufficiently large $z$, for all $\ell'$, the probability that \threephase{} chooses an arrival among the top $\ell'$ is at least 
    \[
        \Pr{\ALG(\bX) \geq x^{(\ell')}} \geq \alpha(\tau_1, \tau_2, z) \cdot \Pr{\MAX(\bX_+) \geq x^{(\ell')}}.
    \]
\end{corollary}

%%%
\subsection{Joint Competitive Ratio Guarantees} \label{sec:alg-competitive-ratios}

\begin{figure}
    \centering
    \begin{tikzpicture}[x=12cm, y=8cm]
    % Draw axes
    \draw[-{Latex}] (0,0) -- (0.4,0) node[right] {$\beta$};
    \draw[-{Latex}] (0,0) -- (0,0.82) node[above] {$\alpha$};
    
    % Draw ticks on x-axis
    \foreach \x in {0,1/e}
        \draw (\x,0.01) -- (\x,-0.01) node[below] {\x};
    
    % Draw ticks on y-axis
    \foreach \y in {0, 1/e, 1-1/e}
        \draw (0.01,\y) -- (-0.01,\y) node[left] {\y};
        
    % Draw points and labels
    \draw[fill=black] (0,0.745) circle[radius=0.5pt] node[left, black] {$(0,\beta_0)$};
    \draw[fill=black] ({1/exp(1)},{1/exp(1)}) circle[radius=0.5pt] node[right, black] {$(\frac{1}{e},\frac{1}{e})$};
    % Representative point
    \draw[fill=black] (0.313481,0.586078) circle[radius=0.5pt] node[right, black] {$(\gamma,\alpha(\gamma))$};

    % Draw blue line
    \draw[blue, dashed] (0,0.745) -- ({1/exp(1)},{1/exp(1)});

    % draw red curve
    \draw[red, smooth, tension=0.5] plot[smooth] file {tikz_lambert-pts.txt};

    % draw dashed red line
    \draw[red, dashed] (0,0.745) -- (0.230501,0.652724);

    % Draw bounding box
    \draw[gray,dotted] (0,0.745) -- ({1/exp(1)},0.745);
    \draw[gray,dotted] ({1/exp(1)},0.745) -- ({1/exp(1)},0);
    
\end{tikzpicture}
    \caption{
    The lower bound on the Pareto frontier for \threephase{} is derived by jointly optimizing the bounds given by \Cref{thm:alg-secretary-guarantee} and \Cref{thm:alg-prophet-guarantee} (solid red). 
    This enables a randomized strategy which improves everywhere (dashed red) over the naive interpolative strategy (dashed blue).
    }
    \label{fig:frontiers-populated}
\end{figure}
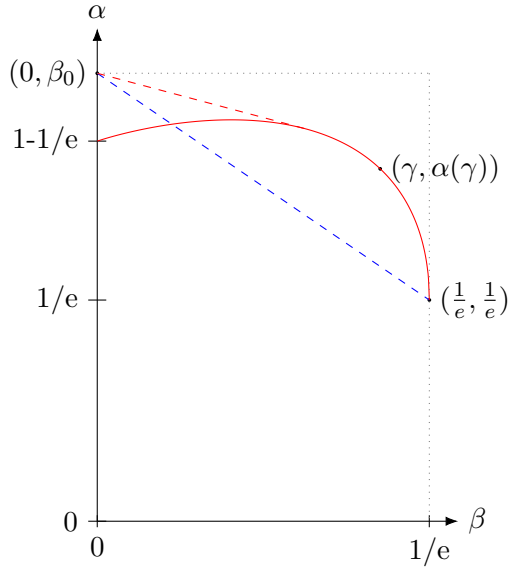

\Cref{lem:prophet-prob-ratio} gives approximate stochastic dominance in the ordinal form of \eqref{eq:approx-stoch-dom-ordered} for the prophet setting, and \Cref{lem:three-phase-prob-ratio-bound} provides something analogous for the top order statistic of $\bX_+$ in the secretary setting.
Our first task is to recover lower bounds on the competitive ratios of \threephase{} in the prophet and secretary settings presented in \Cref{sec:alg-approach}.

For the secretary setting, \Cref{thm:alg-secretary-guarantee} follows from first showing that \Cref{lem:three-phase-prob-ratio-bound} implies approximate stochastic dominance, then proceeding to a competitive ratio guarantee in the usual way.
A proof appears in \Cref{app:algos}.

For the prophet setting, the proof of \Cref{thm:alg-prophet-guarantee} essentially follows from the observation that if $\cD' = \cD$ then all of the decisions of \threephase{} are based on ordinal comparisons between i.i.d.\ arrivals in different phases.
Therefore the order statistic selection probabilities hold independently of the realizations of the order statistic values, and the cumulative probability guarantees translate into expectation guarantees.
Again the proof appears in \Cref{app:algos}.

It remains to establish what joint guarantees we can obtain, and what resulting lower bound on the competitive ratio Pareto frontier we can establish that improves upon the naive interpolation (\Cref{fig:frontiers-unpop}).
By our choice of $\tau_1$ and $\tau_2$ in \Cref{eq:ass-tau1-tau2}, and by \Cref{thm:alg-secretary-guarantee}, we may fix a secretary competitive ratio $\gamma$ and then read off a simultaneous prophet guarantee from \Cref{thm:alg-prophet-guarantee} for \threephase{}. 
This is the red curve in \Cref{fig:frontiers-populated}, plotted in the limit as $z \rightarrow \infty$.

However since this prophet guarantee approaches $1 - \frac{1}{e} < \beta_0 \approx 0.745$ as $\gamma$ approaches $0$, our guarantee for \threephase{} only dominates mixtures of pure secretary-optimal and i.i.d.-prophet-optimal strategies for large values of $\gamma$. In order to improve the whole frontier, for small $\gamma$ we may again interpolate between \threephase{} at some critical value of $\gamma$, and the worst-case optimal prophet strategy which gives $(0, \beta_0)$.
This interpolation is shown in dashed red in \Cref{fig:frontiers-populated}.

\section{The Impossibility of Simultaneous Optimality} \label{sec:hardness}

Having shown that the strategy of randomizing between the prophet and secretary policies can be improved upon, one might wonder if there is even a trade-off between the two settings. In this section, we show that indeed some joint $(\alpha,\beta)$ guarantees are unattainable. 

\begin{theorem}\label{thm:negative-result}
    For any (possibly adaptive) policy, and in the limit as the rate $\lambda$ becomes large, the consistency-robustness guarantees must satisfy
    \begin{align*}
        \alpha \leq \beta_0 - \frac{1}{76}\left(\eps \delta  -\frac{15}{4} \eps^3 \right)^3
        \qquad \text{and} \qquad 
        \beta \leq 1 - (1 - \epsilon)\left(\delta^{\frac{\delta}{1-\delta}} - \delta^{\frac{1}{1-\delta}}\right)\notag
    \end{align*}
    for all $\epsilon \in [0,1/8]$ and $\delta > 0$.
\end{theorem}

This theorem is the direct consequence of \Cref{lem:prophet-must-stop-early,lem:sec-cant-stop-early} in \Cref{sec:poisson} below.

%%%
\subsection{Poisson Hardness Implies \texorpdfstring{$n$}{n}-Arrival Hardness} \label{sec:hardness-transfer}

We establish our joint impossibility in the Poisson setting for the sake of analysis, since certain conditional expectations and event probabilities are more tractable to work with under Poisson arrivals.
However we now show that this choice is without loss of generality.
In the limit as the number of arrivals becomes large, impossibility in the Poisson arrivals setting implies impossibility in the fixed $n$ setting.

By \Cref{lem:poisson-imposs-implies-fixed-n-imposs}, our impossibility (\Cref{thm:negative-result}) transfers to the fixed $n$ i.i.d.\ arrival setting:
\begin{theorem}[Joint impossibility for fixed arrivals]\label{thm:negative-result-fixed-arrivals}
    For any (possibly adaptive) policy, and in the limit as the number of arrivals $n$ becomes large, the consistency-robustness guarantees must satisfy
    \begin{align*}
        \alpha \leq \beta_0 - \frac{1}{76}\left(\eps \delta  -\frac{15}{4} \eps^3 \right)^3
        \qquad \text{and} \qquad 
        \beta \leq 1 - (1 - \epsilon)\left(\delta^{\frac{\delta}{1-\delta}} - \delta^{\frac{1}{1-\delta}}\right)\notag
    \end{align*}
    for all $\epsilon \in [0,1/8]$ and $\delta > 0$.
\end{theorem}

The remainder of this section is dedicated to establishing impossibility in the Poisson setting.

%%%
\subsection{Impossibility for the Known i.i.d.\ Prophet Inequality} \label{sec:poisson}

Before coming to the proof of the consistency-robustness tradeoff, let us first discuss the limits of policies in the known-i.i.d.\ setting. That is, what is possible given perfect distributional knowledge.

A well-known and easily derived fact is the following structure of optimal policies.

\begin{restatable}{lemma}{optimalpolicystructure}
% \begin{lemma}[Structure of the optimal policy] 
\label{lem:opt-policy}
Consider any policy. Let $r(t)$ be the expected reward of this policy when started at time $t$. Then for the function $r\colon [0, 1] \to \mathbb{R}$,
\[
    r'(t) \geq - \int_{r(t)}^\infty \left(1 - \CDF(x)\right) dx.
\]
Moreover, for an optimal policy, this holds with equality and it accepts an arrival of value $x$ at time $t$ if $x > r(t)$, rejects it if $x < r(t)$, and behaves arbitrarily if $x = r(t)$.
% \end{lemma}
\end{restatable}

Relying on this characterization, one can show that one cannot achieve competitive ratios better than $\beta_0 \approx 0.745$. For our argument, we will rely on a hard instance defined by \cite{liu21variable} in the $n$-arrival setting, which we now translate to the Poisson setting with a finite rate.

To state it, let $\beta_n$ be the solution to
\begin{equation} \label{eq:bn-defn}
    \int_{e^{-n}}^1 ((\beta_n^{-1}-1) - y( \log y - 1))^{-1} dy = 1.
\end{equation}
Define $\tilde{y}$ as the unique solution to the following differential equation:
\begin{equation}\label{eq:y-tilde-defn}
    \tilde{y}'(t) = \tilde{y}(t) (\ln \tilde{y}(t) - 1) - (\beta_n^{-1} - 1) \qquad \tilde{y}(0) = 1.
\end{equation}
Note that $\tilde{y}$ is decreasing with $\tilde{y}(0) = 1$ and $\tilde{y}(1) = e^{-n}$. Furthermore, define $\tilde{r}^\ast\colon [0, 1] \to \mathbb{R}$ by 
\begin{equation} \label{eq:r-defn}
    \tilde{r}^\ast(t) \defeq - \int_t^1 \frac{1}{\tilde{y}'(s)} ds. 
\end{equation}
This function is also decreasing and it holds that $\tilde{r}^\ast(1) = 0$. Therefore, the inverse $(\tilde{r}^{\ast})^{-1}$ is a well-defined decreasing function mapping $[0, \tilde{r}^\ast(0)]$ to $[0, 1]$.

We define the distribution of arriving values to be
\begin{equation} \label{eq:hard-dist-CDF}
    \CDF(x) \defeq \begin{cases} 0 & \text{ for $x < 0$} \\
    1 + \frac{1}{n} \ln \tilde{y}((\tilde{r}^\ast)^{-1}(x)) & \text{ for $0 \leq x \leq \tilde{r}^\ast(q)$} \\
    1 + \frac{1}{n} \ln \tilde{y}(q) & \text{ for $\tilde{r}^\ast(q) < x \leq H$} \\
    1 & \text{ for $x > H$}
    \end{cases}
\end{equation}
where $H:= \frac{1}{\tilde y'(q) \ln \tilde y(q)} + \tilde{r}^\ast(q)$.
Note that this is a feasible choice of a CDF. The function $\ln \circ \tilde{y} \circ (\tilde{r}^\ast)^{-1}$ is increasing.
Furthermore, it is bounded below by $-n$ and bounded above by $0$, and so $\CDF(x)$ is an increasing function with range $[0,1]$, regardless of $q \in [0,1]$.
By this definition, we also have
\begin{equation}\label{eq:CDFdefn}
    \CDF(\tilde{r}^\ast(t)) = 1 + \frac{1}{n} \ln \tilde{y}(t).
\end{equation}

\begin{theorem} \label{thm:kertz-hardness-n-asymptotic}
    For the instance \eqref{eq:hard-dist-CDF}, the optimal policy has a competitive ratio of $\beta_0 + O(e^{-n})$. 
\end{theorem}

Our analysis of these instances proceeds in two parts.
We first show that the competitive ratio of the optimal policy approaches $\beta_n$, which is defined for a fixed rate in \eqref{eq:bn-defn}.
We then confirm that $\lim_{n \rightarrow \infty} \beta_n = \beta_0$ and that it converges at an exponential rate, which can also be derived from \cite[Theorem 3.2]{allaart07prophet}.
These together give a bound on the rate at which the optimal competitive ratios of these hard instances approach $\beta_0$. 
Omitted proofs appear in \Cref{app:poisson}.

%%%
\subsection{Good Prophet Policies Must Stop on Early High Values}

To show impossibilities, we will reason about the hard instances introduced in \eqref{eq:hard-dist-CDF} and the behavior of stopping rules that are nearly optimal for these hard instances.

\begin{restatable}{lemma}{lemprophetmuststopearly}
\label{lem:prophet-must-stop-early}
    Let $\epsilon \in [0,1/8]$ and $\delta > 0$. 
    Consider instance \eqref{eq:hard-dist-CDF}.
    Consider any policy that has the following property: Conditioned on there being an arrival of value at least $b < \tilde{r}^\ast(0)$ until time $\epsilon$, it stops on the first such value with probability at most $1 - \delta$.
    Then for sufficiently large rate $\lambda$ its competitive ratio is at most
\[
    \beta_n - \frac{1}{76}\left(\eps \delta  -\frac{3}{2} \eps^3 \left( \delta + \frac{3}{2} \right)\right)^3. 
\]
\end{restatable}

%%%
\subsubsection{Bounds that Hold for all Distributions}

The proof of \Cref{lem:prophet-must-stop-early} will be provided in \Cref{sec:proof-of-prophet-must-stop-early}. 
We first prove two very useful lemmas that hold regardless of distribution, lower-bounding the effect of being constrained not to stop early.

\begin{lemma} \label{lem:proph-imposs-CR-bound}
    Let $\tilde{r}^\ast(\epsilon) < b < \tilde{r}^\ast(0)$. Consider any policy. Let $p(t)$ be the probability that the policy does not accept an arrival of value at least $b$ at time $t$, conditional on there being such an arrival and not having stopped yet, and define $P \defeq \int_0^\epsilon p(t) dt$. Then its expected reward is no more than $\max\left\{ a, \: \tilde{r}^\ast\left(\frac{\kappa P (1-\gamma)}{2 (1 + \kappa)}\right)\right\}$, where $a \defeq \tilde{r}^\ast(\gamma P)$ 
    and 
    \[
        \kappa \defeq \frac{(b - a) \cdot (1 - \CDF(b))}{\int_{\rts(\eps)}^\infty (1-\CDF(x)) dx}.
    \]
\end{lemma}

This lemma is the core of our argument. It can be understood as follows. If $P > 0$, then parameters can be chosen so that the expected reward of the constrained policy is no more than $\tilde{r}^\ast(\epsilon')$ for some value $\epsilon' > 0$. Note that $\tilde{r}^\ast(\epsilon')$ is the expected reward of an optimal policy started at time $\epsilon'$. This means the constraining the policy effectively makes it lose the time from $0$ to $\epsilon'$ to make selections.

\begin{proof}[Proof of \Cref{lem:proph-imposs-CR-bound}]
    Consider any policy fulfilling the constraint. Let $\check{r}(t)$ be the expected reward if we start the policy at time $t$. Our goal is to bound $\check{r}(0)$. As $\check{r}$ is a decreasing function, we can consider its inverse $\check{r}^{-1}$ and compare this to the inverse of $\tilde{r}^\ast$, denoted by $(\tilde{r}^\ast)^{-1}$.

    If $\check{r}(0) \leq a$, there is nothing to be shown. 
    So we may assume for what follows that $\check{r}(0) > a$, and consequently $\check{r}^{-1}(a)$ is well-defined.

    We will show that
    \begin{equation} \label{eq:hardness-proph-inverse-gap-from-eps-to-r(a)}
    \check{r}^{-1}(a) \leq (\tilde{r}^\ast)^{-1}(a) - \frac{\kappa P (1-\gamma)}{2 (1 + \kappa)},
    \end{equation}
    and that for all $s \geq a$,
    \begin{equation} \label{eq:hardness-proph-inverse-gap-from-a-to-s}
        \check{r}^{-1}(a) - \check{r}^{-1}(s) \geq (\tilde{r}^\ast)^{-1}(a) - (\tilde{r}^\ast)^{-1}(s).
    \end{equation}
    Then taking $s = \rc(0)$, \eqref{eq:hardness-proph-inverse-gap-from-eps-to-r(a)} and \eqref{eq:hardness-proph-inverse-gap-from-a-to-s} together imply
    \begin{equation} \label{eq:hardness-proph-gap conclusion}
        (\tilde{r}^\ast)^{-1}(\check{r}(0)) \geq \check{r}^{-1}(\check{r}(0)) + (\tilde{r}^\ast)^{-1}(a) - \check{r}^{-1}(a) \geq \frac{\kappa P (1-\gamma)}{2 (1 + \kappa)}.
    \end{equation}
    By applying $\rts$ to both sides we then obtain $\check{r}(0) \leq \tilde{r}^\ast(\frac{\kappa P (1 - \gamma)}{2(1 + \kappa)})$, since $\rts$ is decreasing. 
    It remains only to justify these key claims. 

    \medskip
    \noindent\textbf{Proof of \eqref{eq:hardness-proph-inverse-gap-from-a-to-s}.} 
    Note that for all $t \in [0, 1]$,
    \[
    (\tilde{r}^\ast)'(t) = - \int_{\tilde{r}^\ast(t)}^\infty (1 - \CDF(x)) dx \qquad \check{r}'(t) \geq - \int_{\check{r}(t)}^\infty (1 - \CDF(x)) dx.
    \]
    Then by the inverse function theorem, for all $s \in [0, \check{r}(0)]$
    \begin{equation}
    \label{eq:hardness-prophet-inv-func-thm-bound-weak}
    (\check{r}^{-1})'(s) = \frac{1}{\check{r}'(\check{r}^{-1}(s))} \leq \frac{1}{- \int_s^\infty (1 - \CDF(x)) dx} = \frac{1}{(\tilde{r}^\ast)'((\tilde{r}^\ast)^{-1}(s))} = ((\tilde{r}^\ast)^{-1})'(s).
    \end{equation}
    Therefore for $a < s$, 
    \[
    \check{r}^{-1}(a) - \check{r}^{-1}(s) = - \int_a^s (\check{r}^{-1})'(s') ds' \geq - \int_a^s ((\tilde{r}^\ast)^{-1})'(s') ds' = (\tilde{r}^\ast)^{-1}(a) - (\tilde{r}^\ast)^{-1}(s).
    \]

    \medskip
    \noindent\textbf{Proof of \eqref{eq:hardness-proph-inverse-gap-from-eps-to-r(a)}.}
    We begin by repeating the argument from \Cref{lem:opt-policy} for the optimal $P$-constrained policy $\check{\tau}$ with threshold $\rc$.
    Informally, our goal will be to show that the slope of $\rc(t)$ is boundedly less negative than that of $\rts(t)$ on the interval $t \in [\rc^{-1}(a), \epsilon]$.
    
    For an infinitesimal time window $\eps$, the probability of a value at least $x$ arriving is $n \cdot \eps \cdot (1-\CDF(x))$. 
    The $P$-constrained optimal policy $\check{\tau}$ accepts as many arrivals above $x = \rc(t)$ in this window as possible, subject to $P$.
    It accomplishes this via some $p(t)$.
    Given $p(t)$, $\check{\tau}$ will opt to reject the lowest possible values above $b$.
    Therefore it accepts (i) all arrivals in $[\rc(t), b)$, and (ii) all arrivals above some threshold $b_{p(t)} > b$.
    Then for $t \in [\rc^{-1}(a), \eps]$ we have
    \begin{align*}
        \check{r}'(t) & = - \int_{\rc(t)}^\infty \prob{\check{\tau}(t) \text{ accepts a value } \geq x \:\vert\: \text{there is an arrival}} \:dx \\
        &= - \int_{\rc(t)}^\infty n\cdot (\CDF(b) - \CDF(x)) + n \cdot (1 - \CDF(b_{p(t)})) \:dx.
        \intertext{By supposing the rejected values are only $b \leq b_{p(t)}$, we only improve the performance $\check{\tau}$ at time $t$ and steepen its derivative:}
        & \geq - \int_{\check{r}(t)}^\infty n \cdot (1 - \CDF(x)) \: dx + \int_{\check{r}(t)}^b  p(t) \cdot n \cdot (1 - \CDF(b)) \: dx \\
        & = - \int_{\check{r}(t)}^\infty n \cdot (1 - \CDF(x)) \: dx + (b - \rc(t))  p(t) \cdot n \cdot (1 - \CDF(b)) \: dx \\
        & \geq - \int_{\check{r}(t)}^\infty n \cdot (1 - \CDF(x)) \: dx + (b - a) \cdot p(t) \cdot n \cdot (1 - \CDF(b)) \: dx,
    \end{align*}
    since $\rc(t) \leq a \leq b$ on $t \in [\rc^{-1}(a), \eps]$.

    By the definition of $\kappa$, we have
    \[
    \check{r}'(t) \geq - (1 - \kappa \cdot p(t)) \int_{\check{r}(t)}^\infty n \cdot (1 - \CDF(x)) dx \geq - \frac{1}{1 + \kappa \cdot p(t)}\int_{\check{r}(t)}^\infty n \cdot (1 - \CDF(x)) dx.
    \]
    Therefore, by the inverse function theorem, we have for all $s \in (\tilde{r}^\ast(\epsilon), a)$
    \begin{equation} \label{eq:hardness-prophet-inv-func-thm-bound}
        (\check{r}^{-1})'(s) = \frac{1}{\check{r}'(\check{r}^{-1}(s))} \leq \frac{1 + \kappa \cdot p(\check{r}^{-1}(s))}{- \int_s^\infty n \cdot (1 - \CDF(x)) dx} = \frac{1 + \kappa \cdot p(\check{r}^{-1}(s))}{(\tilde{r}^\ast)'((\tilde{r}^\ast)^{-1}(s))} = (1 + \kappa \cdot p(\check{r}^{-1}(s))) \cdot ((\tilde{r}^\ast)^{-1})'(s).
    \end{equation}
    Integrating over $s$ from $\check{r}(\epsilon)$ to $a$, we obtain
    \begin{align*}
    \epsilon - \check{r}^{-1}(a) & = - \int_{\check{r}(\epsilon)}^a (\check{r}^{-1})'(s) ds \\
    & \geq - \int_{\check{r}(\epsilon)}^a (1 + \kappa \cdot p(\check{r}^{-1}(s))) \cdot ((\tilde{r}^\ast)^{-1})'(s) ds \\
    & = (\tilde{r}^\ast)^{-1}(\epsilon) - (\tilde{r}^\ast)^{-1}(a) - \kappa \int_{\check{r}(\epsilon)}^a p(\check{r}^{-1}(s)) \cdot ((\tilde{r}^\ast)^{-1})'(s) \:ds\\
    & \geq \epsilon - (\tilde{r}^\ast)^{-1}(a) - \kappa \int_{\check{r}(\epsilon)}^a p(\check{r}^{-1}(s)) \cdot ((\tilde{r}^\ast)^{-1})'(s) \:ds,
    \end{align*}
    where we used that $\rc(\eps) \leq \rts(\eps)$, and so $(\rts)^{-1}(\rc(\eps)) \geq \eps$. 

    In order to lower bound $- \int_{\check{r}(\epsilon)}^a p(\check{r}^{-1}(s)) \cdot ((\tilde{r}^\ast)^{-1})'(s) ds$, let $S = \{ s \mid (\check{r}^{-1})'(s) < 2 ((\tilde{r}^\ast)^{-1})'(s)\}$. 
    As $p(t) \in [0, 1]$ for all $t$, we have
    \begin{align*}
    & - \int_{\check{r}(\epsilon)}^a p(\check{r}^{-1}(s)) \cdot ((\tilde{r}^\ast)^{-1})'(s) ds \\
    & \geq - \frac{1}{2} \int_{\check{r}(\epsilon)}^a p(\check{r}^{-1}(s)) \cdot (\check{r}^{-1})'(s) ds + \frac{1}{2} \int_{\check{r}(\epsilon)}^a p(\check{r}^{-1}(s)) (\check{r}^{-1})'(s) \mathbf{1}_{s \in S} ds \\
    & \geq - \frac{1}{2} \int_{\check{r}(\epsilon)}^a p(\check{r}^{-1}(s)) \cdot (\check{r}^{-1})'(s) ds + \frac{1}{2} \int_{\check{r}(\epsilon)}^a (\check{r}^{-1})'(s) \mathbf{1}_{s \in S} ds.
    \end{align*}
    To bound the first integral, we substitute $\check{r}^{-1}(s)$ by $t$ and obtain because $p(t) \in [0, 1]$
    \[
    - \int_{\check{r}(\epsilon)}^a p(\check{r}^{-1}(s)) \cdot (\check{r}^{-1})'(s) ds = \int_{\check{r}(a)}^{\epsilon} p(t) dt \geq P - \check{r}(a).
    \]
    To bound the second integral, we use that
    \begin{align*}
    & - \int_{\check{r}(\epsilon)}^a (\check{r}^{-1})'(s) \mathbf{1}_{s \in S} ds = - \int_{\check{r}(\epsilon)}^a (\check{r}^{-1})'(s) ds + \int_{\check{r}(\epsilon)}^a (\check{r}^{-1})'(s) \mathbf{1}_{s \not\in S} ds = \epsilon - \check{r}^{-1}(a) + \int_{\check{r}(\epsilon)}^a (\check{r}^{-1})'(s) \mathbf{1}_{s \not\in S} ds \\
    & \stackrel{\eqref{eq:hardness-prophet-inv-func-thm-bound-weak}}{\leq} \epsilon - \check{r}^{-1}(a) + \int_{\check{r}(\epsilon)}^a ((\tilde{r}^\ast)^{-1})'(s) \mathbf{1}_{s \not\in S} ds = \epsilon - \check{r}^{-1}(a) - (\epsilon - (\tilde{r}^\ast)^{-1}(a)) - \int_{\check{r}(\epsilon)}^a ((\tilde{r}^\ast)^{-1})'(s) \mathbf{1}_{s \in S} ds \\
    & \leq (\tilde{r}^\ast)^{-1}(a) - \check{r}^{-1}(a) - \frac{1}{2} \int_{\check{r}(\epsilon)}^a (\check{r}^{-1})'(s) \mathbf{1}_{s \in S} ds,
    \end{align*}
    and therefore
    \begin{align*}
    & - \left( 1 - \frac{1}{2} \right) \int_{\check{r}(\epsilon)}^a (\check{r}^{-1})'(s) \mathbf{1}_{s \in S} ds \leq \check{r}^{-1}(a) - (\tilde{r}^\ast)^{-1}(a).
    \end{align*}
    So
    \[
    \int_{\check{r}(\epsilon)}^a (\check{r}^{-1})'(s) \mathbf{1}_{s \in S} ds \geq - 2 (\check{r}^{-1}(a) - (\tilde{r}^\ast)^{-1}(a)).
    \]

    So, in combination, we obtain
    \begin{align*}
        (\tilde{r}^\ast)^{-1}(a) - \check{r}^{-1}(a) &\geq \frac{\kappa}{2}(P - \check{r}(a)) - \kappa((\tilde{r}^\ast)^{-1}(a) - \check{r}^{-1}(a)), \\
        (1 + \kappa)((\tilde{r}^\ast)^{-1}(a) - \check{r}^{-1}(a)) &\geq \frac{\kappa}{2}(P - \check{r}^{-1}(a)).
    \end{align*}
    Finally applying the definition of $a$, we get
    \[
    \check{r}^{-1}(a) \geq (\tilde{r}^\ast)^{-1}(a) - \frac{\kappa}{2(1 + \kappa)} ( P - (\tilde{r}^\ast)^{-1}(a)) = (\tilde{r}^\ast)^{-1}(a) - \frac{\kappa (1 - \gamma) P}{2(1 + \kappa)}.
    \]
    This proves \eqref{eq:hardness-proph-inverse-gap-from-eps-to-r(a)}, and the lemma claim follows as outlined above.
\end{proof}

In order to prove \Cref{lem:proph-imposs-CR-bound}, we have to move from conditional to unconditional probabilities. 
This move is accomplished in the following lemma.

\begin{restatable}{lemma}{lemprophimpossprobabilities}
\label{lem:proph-imposs-probabilities}
    Consider any policy that has the following property: Up to time $\eps$, it never stops on values below $c$. Furthermore, conditioned on there being an arrival of value at least $c < b < \tilde{r}^\ast(0)$ before time $\epsilon$, it stops on the first such value with probability at most $1 - \delta$. As in \Cref{lem:proph-imposs-CR-bound}, let $p(t)$ be the probability that the policy does not accept an arrival of value at least $b$ at time $t$, conditional on there being such an arrival and not having stopped yet, and define $P = \int_0^\epsilon p(t) dt$. Then
    \[
    P \geq \epsilon \left( 1 - \frac{1 - \delta + \frac{1}{2} \epsilon n (1 - \CDF(b))}{e^{-n \epsilon (1 - \CDF(c))}} \right).
    \]
\end{restatable}
In particular, note that the choice of $c=\tilde{q}^*(\eps)$ satisfies the first assumption of \Cref{lem:proph-imposs-probabilities} without loss of generality, where $\tilde{q}^*(t)$ is the threshold of the optimal policy for a given distribution.

The proof of this lemma has two components. One is from moving from conditional to unconditional probabilities. The other one is shifting attention to the probability of accepting the first arrival above $b$, as opposed to any such arrival. For this change, we rely on the fact that these arrivals come from a Poisson process. If $\epsilon$ is small enough, more than one such arrival is very unlikely.

\subsubsection{Implications for the Hard Prophet Instance}
\label{sec:proof-of-prophet-must-stop-early}

With \Cref{lem:proph-imposs-CR-bound} and \Cref{lem:proph-imposs-probabilities} in hand, we can prove \Cref{lem:prophet-must-stop-early} by plugging in the definitions of instance \eqref{eq:hard-dist-CDF}. It then only remains to bound the function $\ty(t)$ appropriately. To this end, we use the following linear approximations.

\begin{restatable}{observation}{obstildeyderivativebounds}
\label{obs:tilde-y-derivative-bounds}
    For any rate $n$, the function $\tilde{y}(t)$ satisfies 
    \begin{align*}
        - \frac{1}{\beta_n} &\leq \tilde{y}'(t) \leq - \frac{1}{\beta_n} + t \cdot \left(1 + \frac{2}{e} - \frac{1}{\beta_n}\right),  \qquad &t \in [0,\ty^{-1}(1/e)],\\
        1 - \frac{1}{\beta_n} \cdot t &\leq \tilde{y}(t) \:\leq 1 - \frac{1}{\beta_n} \cdot t + \frac{t^2}{2} \cdot \left(1 + \frac{2}{e} - \frac{1}{\beta_n}\right), \qquad &t \in [0,\ty^{-1}(1/e)].
    \end{align*} 
\end{restatable}

These approximations can be derived from Taylor's theorem together with \eqref{eq:y-tilde-defn}.

%%%
\subsection{Good Secretary Policies Cannot Stop on Early High Values}

So far, we have shown that policies that perform well on instance \eqref{eq:hard-dist-CDF} need to have a high probability of stopping early, conditional on there being an early high arrival. Next we turn to the setting in which our distributional knowledge was indeed incorrect, and we show that the high conditional probability of stopping can lead to a poor performance on a different distribution.

\begin{lemma} \label{lem:sec-cant-stop-early}
Let $\epsilon, \delta > 0$. Consider instance \eqref{eq:hard-dist-CDF} and any policy that has the following property: Conditioned on there being an arrival of value at least $b < \tilde{r}^\ast(0)$ before time $\epsilon$, it stops on instance \eqref{eq:hard-dist-CDF} on the first such value with probability at least $1 - \delta$. Then there is a different distribution such that the competitive ratio of this policy is at most
\[
    1 - (1 - \epsilon)\left(\delta^{\frac{\delta}{1-\delta}} - \delta^{\frac{1}{1-\delta}}\right).
\]
\end{lemma}

We first show a useful fact.

\begin{lemma}
Let $f\colon [0, 1] \to [0, 1]$ be any function, and $g\colon [0, 1] \to [0, 1]$ a differentiable non-increasing function. If $\int_0^1 f(x) \:dx \geq 1 - \delta$, then 
\begin{equation} \notag
    \int_0^1 f(x) g(x) \:dx \geq \int_\delta^1 g(x) \:dx.
\end{equation}
\end{lemma}

\begin{proof}
    Note that since $f(x) \leq 1$ for all $x$, we have $\int_0^x f(x') dx' = \int_0^1 f(x') dx' - \int_x^1 f(x') dx' \geq (1 - \delta) - (1 - x) = x - \delta$. So, since $g'(x) \leq 0$ for all $x$, applying integration by parts twice gives us
    \begin{align*}
        \int_0^1 f(x) g(x) dx & = g(1) \int_0^1 f(x) dx - \int_0^1 g'(x) \int_0^x f(x') dx' dx \\
        & \geq g(1) (1-\delta) - \int_0^1 g'(x) \cdot \max(0, x - \delta) dx  = \int_\delta^1 g(x) dx. \qedhere
    \end{align*}
\end{proof}

\begin{proof}[Proof of \Cref{lem:sec-cant-stop-early}]
Consider any fixed policy, and let $f(t)$ be the probability that it stops conditional on the first value above $b$ arrives at time $t$. We have
\[
    \int_0^\epsilon f(t) \frac{\lambda' e^{- \lambda' t}}{1 - e^{- \lambda' \epsilon}} dt \geq 1 - \delta,
\]
where $\lambda' > 0$ is the arrival rate of values above $b$.
(This is because the arrival time of the first arrival of a Poisson process follows an exponential distribution, and $1-e^{-\lambda' \eps}$ is the probability that there is a large arrival in $[0,\eps]$.)
Substituting $x \defeq \frac{1 - e^{- \lambda' t}}{1 - e^{- \lambda' \epsilon}}$, we have
\[
    \int_0^1 f\left( t(x) \right) dx \geq 1 - \delta,
\]
where $t(x) \defeq - \frac{1}{\lambda'} \ln \left( 1 - (1 - e^{- \lambda' \epsilon}) x \right)$.

In the modified instance, values above $b$ arrive with rate $\alpha \lambda'$. The (unconditional) probability of our policy stopping before time $\epsilon$ because of a high arrival is given by %\tk{Explain this!}
\begin{align*}
\int_0^\epsilon f(t) \alpha \lambda' e^{- \alpha \lambda' t} dt & =  \alpha (1 - e^{- \lambda' \epsilon}) \int_0^1 f(t(x)) e^{- (\alpha - 1) \lambda' t(x)} dx \\
& \geq \alpha (1 - e^{- \lambda' \epsilon}) \int_\delta^1 e^{- (\alpha - 1) \lambda' t(x)} dx \\
& = \int_{t(\delta)}^\epsilon \alpha \lambda' e^{- \alpha \lambda' t} dt \\
& = e^{- t(\delta) \alpha \lambda'} - e^{- \epsilon \alpha \lambda'} \\
& = \left(1 - \delta (1 - e^{-\lambda' \epsilon}) \right)^\alpha - e^{- \alpha \lambda' \epsilon} \\
& \geq e^{- \alpha \delta \lambda' \epsilon} - e^{-\alpha \lambda' \epsilon},
\end{align*}
where in the last step we used that by concavity $\delta (1 - e^{-\lambda' \epsilon}) \leq 1 - e^{- \delta \lambda' \epsilon}$. 

We introduce arrivals at rate $\lambda''$ whose values are much larger than $\tilde{r}^\ast(0)$. The above calculation gives an upper bound on the probably that such a value is accepted: Namely, we say that it is accepted if the first such value arrives before time $\epsilon$, which happens with probability $1 - e^{-\lambda'' \epsilon}$, or if it arrives after time $\epsilon$ and the policy has not stopped before time $\epsilon$, which happens with probability $(e^{- \lambda'' \epsilon} - e^{-\lambda''})(1 - (e^{- \alpha \delta \lambda' \epsilon} - e^{-\alpha \lambda' \epsilon}))$. So the overall probability is upper-bounded by
\[
1 - e^{-\lambda'' \epsilon} + (e^{- \lambda'' \epsilon} - e^{-\lambda''})(1 - (e^{- \alpha \delta \lambda' \epsilon} - e^{-\alpha \lambda' \epsilon})).
\]

The probability of any such high-valued arrival is $1 - e^{-\lambda''}$. Neglecting any smaller values, this gives us an upper bound on the competitive ratio of
\[
\frac{1 - e^{-\lambda'' \epsilon} + (e^{- \lambda'' \epsilon} - e^{-\lambda''})(1 - (e^{- \alpha \delta \lambda' \epsilon} - e^{-\alpha \lambda' \epsilon}))}{1 - e^{-\lambda''}} = 1 - \frac{(e^{- \lambda'' \epsilon} - e^{-\lambda''})(e^{- \alpha \delta \lambda' \epsilon} - e^{-\alpha \lambda' \epsilon})}{1 - e^{-\lambda''}}.
\]

If we choose $\alpha = \frac{\ln(1/\delta)}{(1-\delta) \lambda' \epsilon}$, then this becomes
\[
e^{- \alpha \delta \lambda' \epsilon} - e^{-\alpha \lambda' \epsilon} = \delta^{\frac{\delta}{1-\delta}} - \delta^{\frac{1}{1-\delta}}.
\]
Furthermore,
\[
\lim_{\lambda'' \to 0} \frac{e^{- \lambda'' \epsilon} - e^{-\lambda''}}{1 - e^{-\lambda''}} = 1 - \epsilon.
\]

This upper-bounds the competitive ratio by
\[
1 - (1 - \epsilon)\left(\delta^{\frac{\delta}{1-\delta}} - \delta^{\frac{1}{1-\delta}}\right).
\]
\end{proof}

\section{Prophet Sensitivity to Distribution Error} \label{sec:smoothness}

Our investigation is framed in terms of providing simultaneous guarantees for a stopping rule with advice distribution $\cD'$: a guarantee when $\cD' = \cD$ is the true distribution from which arrivals are drawn i.i.d., and a guarantee when $\cD' \neq \cD$ is incorrect.
It is natural to ask whether this dichotomy---either $\cD'$ is accurate or it is not---is the right framing by which to study the i.i.d.\ prophet inequality problem, or whether one should instead hope for performance which depends on how accurate $\cD'$ is, and degrades from the i.i.d.\ prophet to the secretary setting as a function of the discrepancy between $\cD$ and $\cD'$.

Here we demonstrate that the i.i.d.\ prophet inequality problem is very sensitive to misspecification of the arrival distribution, and that such robustness guarantees would be necessarily quite weak in this setting.
We make use of our reductions in \Cref{sec:prelims} together with the impossibility result of \citet{correa19prophet} for the i.i.d.\ prophet inequality problem with unknown distribution.
We will prove this for the Poisson setting; as usual, the reductions in \Cref{sec:prelims} then establish the result for both the $n$-arrival and Poisson models.
We depart from the impossibility of \citeauthor{correa19prophet}.
\begin{lemma}[Theorem 2 of \cite{correa19prophet}] \label{thm:n-secretary-hard-CDFS}
    For all $\delta > 0$ there exists an $n_0$ such that for all $n \geq n_0$, for any $n$-arrival stopping rule $\alg$, there exists an unknown distribution $\cD$ such that, for $\bX$ given by $n$ i.i.d.\ samples from $\cD$,
    \[
        \expect{\alg(\bX)} \leq \left(\frac{1}{e} + \delta\right) \cdot \expect{\opt(\bX)}.
    \]
\end{lemma}

Our $n$-to-Poisson reduction (\Cref{lem:good-poisson-implies-good-fixed-n}) then implies the Poisson version of this claim:
\begin{lemma} 
\label{lem:poisson-secretary-hard}
    For all $\delta \geq 0$ there exists a rate $\lambda_0$ such that, for all $\lambda \geq \lambda_0$, for any stopping rule $\alg_\lambda$ there exists an unknown $\cD$ such that, for Poisson arrivals $\bX$ from $\cD$ at rate $\lambda$,
    \[
        \expectover{\bX \sim \cD}{\alg_\lambda(\bX)} \leq \left(\frac{1}{e} + \delta\right) \cdot \expectover{\bX \sim \cD}{\opt_\lambda(\bX)}.
    \]
\end{lemma}
\begin{proof}
    Suppose for the sake of contradiction that such $\left(\frac{1}{e} + \delta\right)$-competitive stopping rules $\alg_\lambda$ exist for some sequence of $\lambda$ that grows without bound.
    For our choice of $\delta > 0$, we will invoke \Cref{lem:good-poisson-implies-good-fixed-n} with $\alpha \defeq \frac{1}{e} + \delta$ and $\eps \leq \delta/5$. 
    Let $n_0$ be the smallest $n$ such that $\eps = \sqrt{\frac{\log n}{n}} \leq \delta/5$, and then choose one such $\lambda \geq \lambda_0 \defeq n_0 \left(1 - \sqrt{\frac{\log n_0}{n_0}}\right)$.
    \Cref{lem:good-poisson-implies-good-fixed-n} then implies that the derived $n$-arrival stopping rule $\widehat \alg_n$ for the corresponding $n(\lambda)$ is $\alpha - 4\eps = \left(\frac{1}{e} + \frac{\delta}{5}\right)$-competitive.
    For large enough $\lambda$, this contradicts \Cref{thm:n-secretary-hard-CDFS} when invoked with slack $\delta \leftarrow \frac{\delta}{5}$.
\end{proof}

We can now prove that even advice distributions with error on the order of $O(\frac{1}{n})$ are nearly useless.

\begin{theorem}
\label{thm:poisson-not-smooth}
    For all $\delta > 0$ there exists a constant $c$ and a rate $\lambda_0$ such that for all $\lambda \geq \lambda_0$, for all stopping rules $\alg_\lambda$, there exists a pair of distributions $\cD, \cD'$ with $d_{TV}(\cD, \cD') \leq \frac{c}{\lambda}$ such that
    \[
        \expectover{\bX \sim \cD}{\alg_\lambda(\bX, \cD')} \leq \left(\frac{1}{e} + \delta\right) \cdot \expectover{\bX \sim \cD}{\opt_\lambda(\bX)}.
    \]
\end{theorem}
\begin{proof}
    The approach is simple: we will choose $\cD$ to be the hard distribution given by \Cref{lem:poisson-secretary-hard}, padded by 0-value arrivals, and the advice distribution to be $\cD'=0$ identically.

    For the specified $\delta$, let $c$ be the minimum Poisson rate given by \Cref{lem:poisson-secretary-hard} for the same choice of $\delta$, and choose $\lambda_0 = c$ also.
    For any $\lambda \geq \lambda_0$ and corresponding $\alg_\lambda$, without loss of generality $\alg_\lambda$ stops only on nonzero arrivals. 
    Let $\alg_c$ be the stopping rule that results from sending $\alg_\lambda$ arrivals at rate $c$, along with 0-arrivals at rate $\lambda - c$.
    Let $\cD_+$ be the hard distribution given by \Cref{lem:poisson-secretary-hard} for this choice of $\delta$, $\lambda_0=c$, and $\alg_c$. 
    Then we form $\cD$ by sampling from $\cD_+$ w.p. $c/\lambda$, and returning $0$ otherwise.
    The choice of $\cD' = 0$ satisfies $d_{TV}(\cD, \cD') \leq \frac{c}{\lambda}$ and is demonstrably useless as advice for $\alg_\lambda$ on $\cD$; therefore its competitive ratio on $\cD$ is the same as that of $\alg_c$ on $\cD_+$.
    This is at most $(\frac{1}{e} + \delta)$ by \Cref{lem:poisson-secretary-hard}, proving the claim.
\end{proof}
The $n$-arrival version of \Cref{thm:poisson-not-smooth} then follows from \Cref{thm:poisson-not-smooth} via the contrapositive of our Poisson-to-$n$ reduction (\Cref{lem:poisson-imposs-implies-fixed-n-imposs}).

In the language of algorithms with advice, this implies that no stopping rule can have a competitive ratio which is too \emph{smooth} as a function of the prediction error.
More precisely, for prediction error bounds $d_{TV}(\cD, \cD') \leq \eps$, suppose some stopping rule $\alg_n$ has competitive ratio at least
\[
    \max_{\substack{\cD, \: \cD' \\ d_{TV}(\cD, \cD') \leq \eps}} \frac{\expectover{\bX \sim \cD}{\alg_n(\bX,\cD')}}{\expectover{\bX \sim \cD}{\opt_n(\bX)}}  \geq \alpha - \xi(\eps),
\]
for constant $\alpha > \frac{1}{e}$ and some error term $\xi(\eps)$ for which $\xi(0) = 0$. 
Then $\xi(\cdot)$ is at best $\Omega(n)$-Lipschitz.

This can also be seen as a multiplicative counterpart to the additive robustness results of \citet{dutting19posted}, at least for the choice of TV-distance as distributional metric.

\section{Other Models and Future Directions} \label{sec:conclusion}

The most immediate avenue of future work is to improve upon our positive and negative results and further constrain the Pareto frontier for this problem.
We conclude by acknowledging a natural extension of our results and identifying possible next directions.

\paragraph{Advice via Samples}
We remark that the three-phase algorithm \threephase{} is particularly amenable to sample-based prophet inequality settings. Here instead of full knowledge of the distribution $\cD$, some number of samples are revealed to the algorithm from $\cD$ but $\cD$ remains unknown. 
Since \threephase{} uses $z\cdot n$ samples from the advice distribution $\cDp$, \Cref{thm:alg-prophet-guarantee} provides guarantees for the setting where only sample access to the advice distribution is given.
In this setting the guarantee from \Cref{lem:three-phase-prob-ratio-bound} holds even if the advice is an adversarial set of values.
    
For the prophet inequality, one sample per distribution is sufficient to match the optimal guarantee \citep{rubinstein20optimal}. 
In the i.i.d.\ setting, $n-1$ samples suffice for a $(1-\frac{1}{e})$-competitive stopping rule \citep{correa19prophet}, and in fact for any $\epsilon > 0$ a competitive ratio of $\beta_0 - \epsilon$ is possible using $O(n)$ samples \citep{rubinstein20optimal}.
By porting this $(\beta_0 - \epsilon)$-competitive sample-based algorithm to the Poisson setting, we could also interpolate between it and \threephase{}, and attain an additive $\epsilon$-approximation for $\epsilon > 0$ to our Pareto frontier in \Cref{fig:frontiers-populated} via $O(n)$ samples from $\cD'$. 

\paragraph{Prophet Secretary}
If we leave the Poisson model and consider strictly $n$ arrivals with distinct distributions, i.e. the non-i.i.d.\ prophet inequality, consistency-robustness questions are uninteresting for reasons outlined in \Cref{app:other-models} when arrival order is fixed. 
But if we have left the Poisson model then the i.i.d.\ prophet problem and the prophet secretary problem are distinct.

While our algorithmic results can be straightforwardly adapted to the $n$-arrival i.i.d.\ prophet problem, this seems less likely for prophet secretary: stopping rules with guarantees proven via approximate stochastic dominance do not satisfy the `best so far' property \citep{correa21prophet}, and a $(1-\frac{1}{e})$-competitive fixed-threshold rule with this property \citep{ehsani18prophet} does not seem amenable to approximate stochastic dominance arguments. 
It therefore remains an intriguing open question what consistency-robustness guarantees can be attained for prophet secretary.

\paragraph{Stochastic Matching} A prominent generalization of the prophet inequality is to online stochastic matching, where not one but multiple goods are allocated to an arriving sequence of buyers drawn i.i.d.\ (say) from a known distribution over buyer types \citep{feldman09online}.
Here the landscape is more complicated in that the optimal known-distribution competitive ratio is not known, as well as by the multiple model choices regarding vertex vs. edge weights and whether distributions over goods/right-hand side vertices are identical or distinct \citep{haeupler11online,brubach20online}.
But work of \citet{kesselheim13optimal} in the random-order model implies that a $\frac{1}{e}$-approximation is possible for the i.i.d. unknown-distribution setting, and the more recent hardness for the unknown-i.i.d.\ prophet inequality implies this is optimal \citep{correa19prophet}.
Nontrivial joint guarantees may be possible for i.i.d.\ stochastic matching as well.

{\footnotesize
    \bibliography{dblp,refs}
}

\newpage 
\appendix 
%%%
\section{Other Settings: Prophet, Secretary, and Poisson}\label{app:other-models}

In the problem of optimal stopping, an unknown sequence of values arrives one-at-a-time, and each arrival must be irrevocably chosen or rejected before the next value arrives.
The offline optimum $\OPT$ is the maximum of the sequence.
In the worst case this problem is intractable, and $\alpha = 1/n$ is the best competitive ratio possible. 
(This is attained by precommitting to an index $i \in [n]$ u.a.r., and is best possible for the correlated instance which extends the 2-arrival hard instance $X_1 = 1$, $X_2 = 1/\eps$ w.p.\ $\eps$ by additional correlated $X_i = 1/\eps^i$ w.p. $\eps^i$, where $X_i > 0$ only when $X_{i-1} > 0$.)
This motivates beyond-worst-case assumptions, among them most prominently the so-called secretary problem and the prophet inequality paradigm.

For the prophet inequality, an adversary chooses a sequence of distributions $\cD_1, \ldots \cD_n$ and the algorithm faces a sequence of independent samples from these distributions in a known or unknown order.
But if the knowledge of $\cD_1, \ldots \cD_n$ is unreliable then this problem is intractable, because if the advice is wrong, the true $\cD_1, \ldots \cD_n$ may deterministically encode a hard instance.
\begin{observation}
    For the prophet inequality, all attainable consistency-robustness guarantees $(\alpha, \beta)$ satisfy $\alpha \leq 1/2$ and $\beta \leq O(1/n)$.
\end{observation} 

However if the arrivals are independent samples from a \emph{shared} underlying $\cD_i = \cD$, or if the values are independent samples from distinct $\cD_i$ but arrive in a uniformly random order (the so-called prophet secretary problem \cite{esfaniari15prophet}), then when the distributional advice is wrong there is still hope: one can always ignore the advice and deploy the optimal stopping rule for the secretary problem, for a robustness guarantee of $\alpha = 1/e$. 

We consider the Poisson setting, where values arrive according to a Poisson point process with rate $\lambda=n$ in the time interval $t\in [0,1]$, and each arrival value is an independent sample from an underlying distribution $\cD$. 
In this setting, the i.i.d. prophet inequality and the prophet secretary problem coincide: $n$ Poisson processes with rate $\lambda=1$ and values drawn from distinct $\cD_i$ is equivalent to one process with rate $n$ and values drawn from the average distribution.

\section{Proofs from \texorpdfstring{\Cref{sec:prelims}}{Section 2}} \label{app:prelims}

We will make use of tail bounds on the number of arrivals in a Poisson process with rate $\lambda$.
\begin{lemma}[Theorem 1 of \citep{cannonne2019poisson}]
\label{lem:poisson-tail-bounds}
	Let $X \sim Poisson(\lambda)$ for $\lambda > 0$. 
	Then for any $z > 0$,
	\begin{align}
		\Pr{X \leq \lambda - z} & \leq \exp\left(-\frac{z^2}{2\lambda} \right), \label{eq:poisson-lower-tail-bound} \\
		\Pr{X \geq \lambda + z} & \leq \exp\left(-\frac{z^2}{2(\lambda+z)} \right). \label{eq:poisson-upper-tail-bound}
	\end{align}
\end{lemma}

\ntopoisson* 

\begin{proof}[Proof of \Cref{lem:good-poisson-implies-good-fixed-n}]
    For this proof we will let $\ALG_\lambda$ and $\ALG_n$ denote the random variable which is the value on which the Poisson arrival and $n$-arrival algorithms stop, respectively.
    Similarly, $\OPT_\lambda$ and $\OPT_n$ denote the maximum value from among Poisson arrivals at rate $\lambda$ and from among $n$ arrivals, respectively.
    Additionally, we will let $\OPT_\lambda^{> n}$ denote the maximum value from among Poisson arrivals at rate $\lambda$, \emph{subsequent to the first $n$ arrivals}.
    
    Our proof will proceed from a pair of claims. 
    We will prove these for $\lambda = n(1-\eps)$ for $\eps = \sqrt{\frac{\log n}{n}}$.
    Since $\alpha \leq 1$, the claim automatically holds for all $n \leq 50$. 
    
    First, we will require that for any $x > 0$, it is much less likely that $\OPT_\lambda^{> n}\geq x$ than that $\OPT_n \geq x$. 
    In particular, via a union bound we have
    \begin{align*}
        \Pr{\OPT_\lambda^{>n} \geq x} &= \sum_{k \geq 1} \Pr{Poisson(\lambda) = n+k}\cdot \Pr{\OPT_k \geq x} \\
        &\leq \sum_{k \geq 1} \Pr{Poisson(\lambda) = n+k}\cdot \left\lceil \frac{k}{n}\right\rceil \Pr{\OPT_n \geq x} \\
        &= \Pr{\OPT_n \geq x} \sum_{\ell \geq 1} \Pr{Poisson(\lambda) \geq \ell\cdot n} \\
        &\leq \Pr{\OPT_n \geq x} \cdot \left( \exp\left(-\frac{(n - \lambda)^2}{2n}\right) + \sum_{\ell \geq 2} \exp\left(-\frac{(\ell n - \lambda)^2}{2\ell n}\right) \right) \\
        &\leq \Pr{\OPT_n \geq x} \cdot \left( n^{-1/2} + e^{-n/6}\cdot\sum_{\ell \geq 0} e^{-\ell n/3} \right) \\
        &\leq 2n^{-1/2} \cdot \Pr{\OPT_n \geq x} 
    \end{align*}
    for $n \geq 7$.
    In the last steps we applied \Cref{lem:poisson-tail-bounds}, our choice of $\lambda = n-\sqrt{n\log n}$, and the fact that $\frac{(\ell n - \lambda)^2}{2 \ell n} \geq \frac{(\ell - 1)^2 n}{2 \ell} \geq n(\frac{\ell}{3}- \frac{1}{2})$.
    Therefore
    \begin{equation} \label{eq:poisson-to-n-suffix-opt-unlikely}
		\Pr{\OPT_\lambda^{>n} \geq x} \leq \frac{2}{\sqrt{n}} \cdot \Pr{\OPT_n \geq x}.
    \end{equation}

	We will also use that at rate $\lambda$, the Poisson optimum almost stochastically dominates the $n$-arrival optimum.
	To show this we couple the distributions of $\bX_n$ and $\bX_\lambda$.
	To take a joint draw $(\bX_n, \bX_\lambda)$,
	\begin{itemize}
		\item Sample $\bX_n = X_1, \ldots, X_n$ by drawing $n$ i.i.d. samples from $\cD$.
		\item Randomly permute $\bX_n$ to form $Y_1, \ldots, Y_n$; sample $Y_i \sim \cD$ i.i.d. for all $i \geq n+1$.
		\item Form $\bX_\lambda$ by sampling Poisson arrival times at rate $\lambda$, and setting the $i$-th arrival value to $Y_i$.
	\end{itemize}
	Now suppose that for a joint draw $(\bX_n, \bX_\lambda)$, it holds that $\max \bX_n \geq x$. 
	What is the probability that $\max \bX_\lambda \geq x$ for the corresponding $\bX_\lambda$?
	Let $X^*$ be an arbitrary value in $\bX_n$ for which $X^* \geq x$. 
	What is the probability $X^*$ makes it into $\bX_\lambda$?
	By \Cref{lem:poisson-tail-bounds} we have for $z > 0$ that
	\begin{align*}
		\Pr{\max \bX_\lambda \geq x \vert \max \bX_n \geq x} &\geq \Pr{X^* \in \bX_\lambda \given X^* \in \bX_n} \notag \\
		&\geq \frac{\lambda - z}{n} \cdot \Pr{Poisson(\lambda) \geq \lambda - z} \notag \\
		&\geq \frac{\lambda - z}{n} \left(1 - e^{-\frac{z^2}{2 \lambda}} \right) \\
		&\geq \left(1 - \frac{2\sqrt{\log n}}{\sqrt{n}}\right) \cdot  \left(1 - \frac{1}{\sqrt{n}} \right) \\
		&\geq 1 - n^{-1/2}\left(1 + 2\log^{1/2} n\right) 
	\end{align*}
    for the choice of $z = \sqrt{n \log n}$ and $\lambda = n - \sqrt{n \log n}$.
	Multiplying by $\Pr{\max \bX_n \geq x}$ then gives
	\begin{align}
		\Pr{\OPT_\lambda\geq x} &\geq \left(1- \frac{1 + 2\sqrt{\log n}}{\sqrt{n}}\right) \cdot \Pr{\OPT_n \geq x}, \notag \\
        \expect{\OPT_\lambda} &\geq \left(1- \frac{1 + 2\sqrt{\log n}}{\sqrt{n}}\right) \cdot \expect{\OPT_n}. \label{eq:poisson-to-n-poisson-opt-almost-n-opt}
	\end{align}

	We are now ready to combine these claims to lower bound the performance of $\ALG_n$.
	We will argue by approximate stochastic dominance.
	For any $x > 0$, if our simulation of $\ALG_\lambda$ attempts to stop on anything after the $n$-th arrival, then $\ALG_n$ will fail.
    Therefore its failure to stop on a given value (relative to $\ALG_\lambda$) is at most the probability of seeing such a value after the $n$-th arrival.
    Applying this, we have
    \begin{align*}
		\Ex{\ALG_n} & = \int_0^\infty \Pr{\ALG_n \geq x} dx \\
        & \geq \int_0^\infty \Pr{\ALG_\lambda \geq x} - \Pr{\OPT_\lambda^{>n} \geq x} dx \\
		& \geq \int_0^\infty \Pr{\ALG_\lambda \geq x} - \frac{2}{\sqrt{n}} \cdot \Pr{\OPT_n \geq x} dx
		&& \text{(by \eqref{eq:poisson-to-n-suffix-opt-unlikely})} \\
		& = \Ex{\ALG_\lambda} - \epsilon \cdot \Ex{\OPT_n} \\
        & \geq \alpha \cdot \Ex{\OPT_\lambda} - \frac{2}{\sqrt{n}} \cdot \Ex{\OPT_n} &&
		\text{(by assumption on $\ALG_\lambda$)} \\
		& \geq \alpha \left(1- \frac{1 + 2\sqrt{\log n}}{\sqrt{n}}\right)\cdot \Ex{\OPT_n} - \frac{2}{\sqrt{n}} \cdot \Ex{\OPT_n} 
		&& \text{(by \eqref{eq:poisson-to-n-poisson-opt-almost-n-opt})} \\
		& \geq (\alpha - 4\eps) \cdot \Ex{\OPT_n}
	\end{align*}
    for all $\alpha$ and $n \geq 10$ for our choice of $\eps$.
\end{proof}

\poissonton*

\begin{proof}[Proof of \Cref{lem:good-fixed-n-implies-good-poisson}]
    Given Poisson arrivals $\bX$ from distribution $\cD$ and a fixed-$n$ stopping rule $\alg_n$, we will construct $\alg_\lambda$ by constructing a sequence $\tilde \bX = \tilde X_1, \ldots, \tilde X_n$ of length $n$ using the Poisson arrivals, and then simulate $\alg_n$ on $\tilde \bX$.
    In particular, let $\tilde \cD$ be the distribution which is $0$ with probability $e^{-\eps}$ and $\cD$ otherwise.
    Then we will show that if $\alg_n$ is $\alpha$-competitive on $n$ i.i.d. arrivals from $\tilde \cD$, 
    then $\alg_\lambda$ is $\alpha\cdot(1-\eps/2)$-competitive on Poisson arrivals from $\cD$ at rate $\lambda = \eps n$.
    The bulk of the argument is then lower bounding the performance of $\alg_\lambda$ on the Poisson arrivals using our assumption that $\alg_n$ is $\alpha$-competitive for the fixed $n$ arrivals setting.

    We first detail the construction of $\tilde \bX$ from the i.i.d. Poisson arrivals at rate $\lambda = \eps \cdot n$ from distribution $\cD$ for $t \in [0,1]$.
    Let $F$ be the CDF of the Poisson arrival distribution $\cD$.
    Let $\tilde X_i$ be the first Poisson arrival in the range $\left[\frac{i-1}{n}, \frac{i}{n}\right)$ if a first arrival exists, and $\tilde X_i = 0$ otherwise.
    Then the CDF $\tilde F$ of this new distribution $\tilde \cD$ satisfies
    \begin{align*}
        1 - \tilde F(x) &= \prob{\text{first arrival in } [0,1/n] \text{ is } \geq x}  \\
        &=\int_0^{1/n} \lambda (1-F(x)) \cdot e^{-\lambda t} \:dt \\
        &=\lambda (1-F(x)) \cdot \frac{1}{\lambda} \left( 1- e^{-\lambda /n } \right) \\
        &= (1-F(x)) \cdot \left( 1- e^{-\lambda /n } \right) \\
        &\geq (1-F(x)) \cdot \eps (1 - \eps/2),
    \end{align*}
    where in the last step we use that $\lambda = \eps n$ and the fact that $e^{-x} \leq 1 - x + x^2/2$ for $x \in [0,1]$.
    Rearranging and letting $\tilde F^\uparrow(x)$ denote the CDF of the maximum of $n$ draws from $\tilde \cD$, we then have 
    \begin{align*}
        \tilde F(x) &\leq 1 - (1-F(x))\cdot \eps (1-\eps/2), \\
        \tilde F^\uparrow(x) &\leq \left(1 - (1-F(x))\cdot \eps (1-\eps/2)\right)^n \\
        &\leq \exp\left( - (1-\eps/2)\cdot \eps n \cdot (1-F(x)) \right),
    \end{align*}
    since $(1-\frac{1}{n} x)^n \leq e^{-x}$ for $n \geq 1$ and $x \in [0,1]$.

    With this in hand, we now turn to bounding the competitive ratio of $\alg_\lambda$.
    Starting with the assumed competitive ratio guarantee for $\alg_n$ on our constructed sequence $\tilde \bX$, we have
    \begin{align*}
        \expectover{\bX \sim \cD}{\alg_\lambda (\bX)} &= \expectover{\tilde \bX \sim \tilde \cD}{\alg_n(\tilde \bX)} \\
        &\geq \alpha \cdot \int_0^\infty \left( 1 - \tilde F ^\uparrow (x)\right) \: dx \\
        &\geq \alpha \cdot \int_0^\infty \left( 1 - \exp\left( - (1-\eps/2)\cdot \eps n \cdot (1-F(x)) \right)\right) \: dx \\
        &\geq \alpha \cdot (1-\eps/2)\cdot \int_0^\infty \left( 1 - \exp\left( - \eps n \cdot (1-F(x)) \right)\right) \: dx,
        \intertext{since $f(x) = 1 - e^{a x}$ is concave for $a \in [0,1]$ and $f(0) = 0$. Finally, since $\lambda = \eps n$ this is the expected maximum of the Poisson process $\bX$, and we have}
        \expectover{\bX \sim \cD}{\alg_\lambda} & \geq \alpha \cdot (1-\eps/2)\cdot \expectover{\bX\sim \cD}{\max(\bX)},
    \end{align*}
    proving the claim.
\end{proof}
\section{Proofs from \texorpdfstring{\Cref{sec:algos}}{Section 3}} \label{app:algos}

\conditionalprobforlowXj*

\begin{proof}[Proof of \Cref{lem:condnl-prob-low-j}]
To show this for fixed $x_i = x^{(j)}$, consider the largest arrival in the range $\bX_1$ and let $k$ be its order statistic, so that $x^{(k)} \defeq \max_{i' \in [i-1]} x_{i'}$. 

Let $\Lambda_i(j)$ denote the event that $x_i = x^{(j)}$ is the largest arrival so far, i.e. that $x^{(k)} \leq x^{(j)}$. 
Let $\neg \Lambda_i(j)$ denote its complement. 

We will also define the events $\Xi_i^1(j)$, $\Xi_i^2(j)$, and $\Xi_i^3(j)$ in order to case on the value of the largest arrival $x^{(k)}$ prior to $x_i = x^{(j)}$, relative to the values $x^{(\ell)}$ and $x_0^{(\ell)}$:
\begin{itemize}
    \item $\Xi_i^1(j)$ corresponding to $x^{(k)} \geq x^{(\ell)}$,
    \item $\Xi_i^2(j)$ corresponding to $x_0^{(\ell)} \leq x^{(k)} < x^{(\ell)}$,
    \item $\Xi_i^3(j)$ corresponding to $x^{(k)} < x_0^{(\ell)}$.
\end{itemize}
Since $x^{(\ell)} \geq x_0^{(\ell)}$, these events partition our probability space.

Then the law of total probability gives
\begin{align}
    &\Pr{\ALG(\bX) = x^{(j)}\given x_i = x^{(j)}} = \notag \\ 
    &\quad \Pr{\ALG(\bX) = x^{(j)}\given \Lambda_i(j) \wedge \Xi_i^1(j) \wedge x_i = x^{(j)}} \cdot \Pr{\Lambda_i(j) \wedge \Xi_i^1(j) \given x_i = x^{(j)}} \label{eq:accept-prob-split} \\
    &\qquad + \Pr{\ALG(\bX) = x^{(j)}\given\Lambda_i(j) \wedge \Xi_i^2(j) \wedge x_i = x^{(j)}} \cdot \Pr{\Lambda_i(j) \wedge \Xi_i^2(j) \given x_i = x^{(j)}} \notag \\
    &\qquad \qquad + \Pr{\ALG(\bX) = x^{(j)}\given\Lambda_i(j) \wedge \Xi_i^3(j) \wedge x_i = x^{(j)}} \cdot \Pr{\Lambda_i(j) \wedge \Xi_i^3(j) \given x_i = x^{(j)}}, \notag 
\end{align}
since by \cref{obs:best-so-far} if $\neg \Lambda_i(j)$ holds then $\ALG(\bX) \neq x^{(j)}$ and so all $\neg \Lambda_i(j)$ terms vanish.

Case 1 and Case 2 below correspond to $\Xi_i^1(j)$ and $\Xi_i^3(j)$. We will disregard the $\Xi_i^2(j)$ event and prove a lower bound on $\Pr{\ALG(\bX) = x^{(j)}\given x_i = x^{(j)}}$. Our analysis will proceed according to the order statistics over both the simulated and actual arrivals $\bX = \bX_0 \cup \bX_+$.

In Case 1 and Case 2 below, we will assume that $x^{(j)} \geq x^{(\ell)}$ (an the assumption of the lemma) and consider the events $\Xi_i^3$ and $\Xi_i^1$, respectively.

\bigskip
%%%
\noindent\textbf{Case 1: $\Lambda_i(j) \wedge \Xi_i^3(j)$, or $x^{(j)} \geq x^{(\ell)} \geq x_0^{(\ell)} \geq x^{(k)}$:}

If $x^{(k)}$ arrives in $\{x_1, \ldots, x_{r_1 - 1}\}$ then nothing smaller will be accepted in Phase 2 or Phase 3 by \cref{obs:best-so-far}, and so $x_i = x^{(j)}$ will be accepted.

If $x^{(k)}$ arrives in $\bX_2$, then $x^{(k)} \leq x_0^{(\ell)}$ and neither it nor any other arrival in $\{x_{r_1}, \ldots, x_{i-1}\}$ will be accepted in Phase 2, and by \cref{obs:best-so-far} no other arrivals will be accepted in Phase 3 prior to the arrival of $x^{(j)}$. 
Therefore $x_i = x^{(j)}$ will be accepted.

Finally, if $x^{(k)}$ arrives in $\{x_{r_2}, \ldots, x_{i-1}\}$, then it---or some arrival prior to it---will be accepted instead, since it beats the Phase 3 threshold by definition. 
Therefore
\begin{align}
    \Pr{\ALG(\bX) = x^{(j)}\given \Lambda_i(j) \wedge \Xi_i^3(j) \wedge x_i = x^{(j)}} &= \Pr{x^{(k)} \text{ arr. before } r_2} = \frac{\min(r_2-1, i-1)}{i-1}.
\end{align}

We also need the probability that $\Lambda_i(j)$ and $\Xi_i^1(j)$ hold; since for $j \leq \ell$ we have $\Xi_i^3(j) \subseteq \Lambda_i(j)$, this is given by
\begin{align*}
    \Pr{\Lambda_i(j) \wedge \Xi_i^3(j) \given x_i = x^{(j)}} & = \Pr{\Xi_i^3(j) \given x_i = x^{(j)}} \\
    &=\Pr{\bigvee_{j' > \ell} \left[x^{(k)} = x^{(j')} \right] \given x_i = x^{(j)}} \\
    & = \sum_{j' \geq \ell+1} \Pr{x^{(k)} = x^{(j')} \given x_i = x^{(j)}} \\
    & = \sum_{j' \geq \ell+1} \Pr{\bigwedge_{j'' < j',\: j'' \neq j} \left[x^{(j'')} \not \in \{x_1, \ldots, x_{i-1}\} \right] \wedge \left[ x^{(j')} \in \{x_1, \ldots, x_{i-1}\} \right]} \\
    & = \sum_{j' \geq \ell+1} \left(\frac{n(z+1) - i}{n(z+1) - 1}\right)^{j' - 2} \frac{i-1}{n(z+1)-1},
\end{align*}
% where the error term comes from the 

Putting this together, for $i \geq r_1$ we have
\begin{align}
    &\Pr{\ALG(\bX) = x^{(j)}\given \Lambda_i(j) \wedge \Xi_i^3(j) \wedge x_i = x^{(j)}} \cdot \Pr{\Lambda_i(j) \wedge \Xi_i^3(j) 
    \given x_i = x^{(j)}} \notag \\
    &\qquad \qquad = \sum_{j' \geq \ell+1} \left(\frac{n(z+1) - i}{n(z+1) - 1}\right)^{j' - 2} \cdot \frac{\min(r_2 - 1, i - 1)}{n(z+1) - 1} \notag \\
    &\qquad \qquad \geq \left( \frac{n(z + 1) - i}{n(z + 1)  - 1} \right)^{\ell - 1} \cdot \sum_{j' = 0}^\infty \left(\frac{n(z + 1) - i}{n(z + 1)  - 1}\right)^{j'} \cdot \frac{\min(r_2 - 1, i - 1)}{n(z+1) - 1} \notag \\
    &\qquad \qquad = \left( \frac{n(z+1) - i}{n(z+1) - 1} \right)^{\ell - 1} \cdot \frac{\min(r_2 - 1, i-1)}{i-1}. \label{eq:case1-prob}
\end{align}

\bigskip
%%%
\noindent\textbf{Case 2: $\Lambda_i(j) \wedge \Xi_i^1(j)$, or $x^{(j)} \geq x^{(k)} \geq x^{(\ell)}$:}

If $x^{(k)} \geq x^{(\ell)}$ and $x^{(k)}$ arrives in $\{x_1, \ldots, x_{r_1 - 1}\}$ then it will not be accepted on arrival and by \cref{obs:best-so-far} $\ALG$ will not accept anything before reaching $x_i = x^{(j)}$. Furthermore for $i \geq r_1$ this $x_i = x^{(j)}$ will be accepted on arrival, since it exceeds both $x^{(j)}$ and $x^{(\ell)}$ and so meets the thresholds in both Phase 2 and Phase 3.

If $x^{(k)} \geq x^{(\ell)}$ and it arrives in $\bX_2$ then it is larger than the largest arrival in Phase 1, and so it (or an arrival in Phase 2 prior to it) is accepted, and so $x_i = x^{(j)}$ is not. 
The same is true if $x^{(k)}$ arrives in $\bX_3$. 
Therefore
\begin{align}
    \Pr{\ALG(\bX) = x^{(j)}\given \Lambda_i(j) \wedge \Xi_i^1(j) \wedge x_i = x^{(j)}} &= \Pr{k \in [r_1 - 1] } = \frac{r_1 - 1}{i - 1}.
\end{align}

Following a similar argument, the probability that $\Lambda_i(j)$ and $\Xi_i^1(j)$ hold is now
\begin{align*}
    \Pr{\Lambda_i(j) \wedge \Xi_i^1(j) \given x_i = x^{(j)}} & = \Pr{\bigvee_{j+1 \leq j' \leq \ell} \left[x^{(k)} = x^{(j')} \right] \given x_i = x^{(j)}} \\
    & = \sum_{j' = j+1}^\ell \Pr{x^{(k)} = x^{(j')} \given x_i = x^{(j)}} \\
    & = \sum_{j' = j+1}^\ell \Pr{\bigwedge_{j'' < j',\: j'' \neq j} \left[j'' \not \in [i-1] \right] \wedge \left[ j' \in [i-1] \right]} \\
    & = \sum_{j' = j+1}^\ell \left(\frac{n(z + 1) - i}{n(z + 1) - 1}\right)^{j' - 2} \cdot \frac{i-1}{n(z+1) - 1}.
\end{align*}

Therefore in this case, for $i \geq r_1$ we have
\begin{align}
    &\Pr{\ALG(\bX) = x^{(j)}\given \Lambda_i(j) \wedge \Xi_i^1(j) \wedge x_i = x^{(j)}} \cdot \Pr{\Lambda_i(j) \wedge \Xi_i^1(j) \given x_i = x^{(j)}} \notag \\
    &\qquad \qquad = \sum_{j' = j+1}^\ell \left(\frac{n(z + 1) - i}{n(z + 1) - 1}\right)^{j' - 2} \cdot \frac{r_1 - 1}{n(z+1)-1} \notag \\
    &\qquad \qquad = \left(\frac{n(z + 1) - i}{n(z + 1) - 1}\right)^{j - 1} \cdot \sum_{j'=0}^{\ell-j-1} \left(\frac{n(z + 1) - i}{n(z + 1) - 1}\right)^{j'} \cdot \frac{r_1 - 1}{n(z + 1) - 1} \notag \\
    &\qquad \qquad = \left(\left(\frac{n(z + 1) - i}{n(z + 1) - 1}\right)^{j - 1} - \left(\frac{n(z + 1) - i}{n(z + 1) - 1}\right)^{\ell-1} \right) \cdot \frac{r_1 - 1}{i - 1}. \label{eq:case2-prob}
\end{align} 

Combining \eqref{eq:case1-prob} and \eqref{eq:case2-prob} with \eqref{eq:accept-prob-split} and rearranging then gives 
\begin{align*}
    &\Pr{\ALG(\bX) \geq x^{(j)}\given x_i = x^{(j)}} \\
    &\qquad\geq \left(1- \frac{i-1}{n(z+1)-1}\right)^{j - 1} \cdot \frac{r_1-1}{i-1}  + \left( 1- \frac{i-1}{n(z+1)-1}\right)^{\ell - 1} \cdot \left( \frac{\min(r_2-1, i-1)}{i-1} - \frac{r_1-1}{i-1}\right), %\label{eq:condnl-prob}
\end{align*}
as claimed.
\end{proof}

\conditionalprobforhighXj*

\begin{proof} [Proof of \Cref{lem:condnl-prob-high-j}]
    Recall \Cref{obs:best-so-far}: \threephase{} only accepts $x^{(j)}$ arriving at step $i$ if it the largest arrival so far; that is, if $\Lambda_i(j)$ holds.
    Since $j > \ell$ by assumption, no prior arrivals will have been accepted in Phase 2.
    So for $j > \ell$, $x^{(j)}$ is accepted if the largest prior arrival $x^{(k)}$ is smaller, and $x^{(k)}$ does not arrive in $\{x_{r_2}, \ldots, x_{i-1}\}$; otherwise $x^{(k)}$ will be accepted instead.
    
    Following a similar argument to those in \Cref{lem:condnl-prob-low-j}, in this case we have
    \begin{align*}
        \Pr{\ALG(\bX) = x^{(j)}\given x_i = x^{(j)}} &= \Pr{\bigvee_{j' \geq j+1} \left[x^{(k)} = x^{(j')} \wedge k \not\in \{r_2, \ldots, i-1\} \right] \given x_i = x^{(j)}} \\
        &\geq \sum_{j' = j+1}^\infty \Pr{x^{(k)} = x^{(j')} \wedge k \in [r_2 -1] \given x_i = x^{(j)}} \\
        &=\sum_{j' = j+1}^\infty \Pr{\bigwedge_{j'' \neq j, j'' < j'} \left[j'' \not \in [i-1] \right] \wedge j' \in [r_2 - 1] \given x_i = x^{(j)}} \\
        &=\sum_{j' = j+1}^\infty \left(\frac{n(z+1)-i}{n(z+1)-1}\right)^{j' - 2} \cdot \frac{r_2 - 1}{n(z+1)-1} \\
        &= \left(1 - \frac{i-1}{n(z+1)-1}\right)^{j-1} \cdot \frac{r_2-1}{i-1}. \qedhere
    \end{align*}
\end{proof}

\prophetProbabilityRatioTop*

\begin{proof}[Proof of \Cref{lem:prophet-prob-ratio-top}]
    We start by observing that the conditional probability of \Cref{lem:condnl-prob-low-j} can be converted to an unconditional probability as follows. 
    By the law of total probability,
    \begin{equation} \label{eq:threephase-prophet-prob-over-time}
        \Pr{\ALG(\bX) = x^{(j)}} 
        = \sum_{i = r_1}^n \Pr{\ALG(\bX) = x^{(j)} \given x_i = x^{(j)}} \cdot \frac {1}{n(z+1)}.
    \end{equation}
    This is because $\ALG$ only accepts for $x_i \in \bX_2 \cup \bX_3$ and because $x^{(j)}$ is uniformly distributed in $\bX$. 

    Then applying \Cref{lem:condnl-prob-low-j}, using the fact that $\left(1- \frac{i-1}{n(z+1)-1}\right) \geq \left(1- \frac{1}{n(z+1)-1}\right)$, letting $Z \defeq n(z+1)-1$, and doing some integrals, we have that for $\ell' \leq \ell$,
    \begin{align}
        &\Pr{\ALG(\bX) \geq x^{(\ell')}} = \sum_{j=1}^{\ell' } \Pr{\ALG(\bX) = x^{(j)}} \notag \\
        &= \sum_{j=1}^{\ell' } \sum_{i = r_1}^n \Pr{\ALG(\bX) = x^{(j)} \given x_i = x^{(j)}} \cdot \frac {1}{Z+1} \: dt \tag{by \eqref{eq:threephase-prophet-prob-over-time}}\\
        &\geq \sum_{j=1}^{\ell' } \frac {1}{Z+1} \sum_{i=r_1}^n \left(1- \frac{i-1}{Z}\right)^{j - 1} \cdot \frac{r_1-1}{i-1}  + \left( 1- \frac{i-1}{Z}\right)^{\ell - 1} \cdot \left( \frac{\min(r_2 -1, i-1)}{i-1} - \frac{r_1-1}{i-1}\right) \: dt \tag{\Cref{lem:condnl-prob-low-j}}.
    \end{align}
    
    We now approximate this inner sum via an integral. 
    Letting $t \approx \frac{i-1}{n}$, recalling $r_1 = \lceil n \tau_1 \rceil$ and $r_2 = \lceil n \tau_2 \rceil$, and observing that $Z/n \approx z+1$, we have for constant $\tau_1>0$ that
    \begin{align}
        &\Pr{\ALG(\bX) \geq x^{(\ell')}} \notag \\
        &\geq \sum_{j=1}^{\ell' } \left( \frac {1}{z+1} \int_{\tau_1}^1 \left(1- \frac{t}{z+1}\right)^{j - 1} \cdot \frac{\tau_1}{t}  + \left( 1- \frac{t}{z+1}\right)^{\ell - 1} \cdot \left( \frac{\min(\tau_2, t)}{t} - \frac{\tau_1}{t}\right) \: dt - O\left(n^{-1}\right) \right) \notag \\
        &\geq \sum_{j=1}^{\ell' } \frac {1}{z+1} \int_{\tau_1}^1 \left(1- \frac{t}{z + 1}\right)^{j - 1} \cdot \frac{\tau_1}{t}  + \left( 1- \frac{t}{z + 1}\right)^{\ell - 1} \cdot \frac{\min(\tau_2, t)}{t} - \left( 1- \frac{1}{z + 1}\right)^{\ell - 1} \frac{\tau_1}{t} \: dt - O\left(\ell' n^{-1}\right) \notag \\
        &\geq \frac {1}{z+1} \sum_{j=1}^{\ell' } \left(1- \frac{1}{z + 1}\right)^{j - 1} \cdot \int_{\tau_1}^1 \frac{\tau_1}{t} dt  \notag\\
        &\qquad\qquad + \frac{\ell'}{z+1} \int_{\tau_1}^{\tau_2} \left( 1- \frac{t}{z + 1}\right)^{\ell - 1} dt 
        + \frac{\ell'}{z+1} \left( 1- \frac{1}{z + 1}\right)^{\ell - 1} \cdot \left( \int_{\tau_2}^1 \frac{\tau_2}{t} - \int_{\tau_1}^1 \frac{\tau_1}{t} dt \right)  
        - O\left(\ell' n^{-1}\right) \notag \\
        % \label{eq:prob-alg-above-ell-LB}
        %
        &= \left(1 - \left(1- \frac{1}{z + 1}\right)^{\ell'} \right) \cdot \tau_1 \log\left( \frac{1}{\tau_1}\right) + \frac {\ell'}{z+1} \int_{\tau_1}^{\tau_2} \left( 1- \frac{t}{z + 1}\right)^{\ell - 1} dt - O\left(\ell' n^{-1}\right), \notag
        \intertext{using that $\int_{\tau_1}^1 \frac{\tau_1}{t} dt = \tau_1 \log\left(\frac{1}{\tau_1}\right) = \alpha = \tau_2 \log\left(\frac{1}{\tau_2}\right) = \int_{\tau_2}^1 \frac{\tau_2}{t} dt$ by our choice of $\tau_1$ and $\tau_2$ (\cref{eq:ass-tau1-tau2}). Then recalling that $\ell-1 = a(z+1)$ by definition (\Cref{ass:a-value}), we have}
        &= \left(1 - \left(1- \frac{1}{z + 1}\right)^{\ell'} \right) \cdot \tau_1 \log\left( \frac{1}{\tau_1}\right) + \frac {\ell'}{z+1} \left(\int_{\tau_1}^{\tau_2} e^{-a t} \:dt - O\left(\frac{a}{z+1}\right)\right) - O\left(\ell' n^{-1}\right),  \notag \\
        &= \left(1 - \left(1- \frac{1}{z + 1}\right)^{\ell'} \right) \cdot \tau_1 \log\left( \frac{1}{\tau_1}\right) + \frac {\ell'}{z+1} \left( \frac{1}{a} \left(e^{-a\tau_1} - e^{-a\tau_2}\right) - O\left(\frac{a}{z+1}\right)\right) - O\left(\ell' n^{-1}\right) , \label{eq:low-l'-alg-bound}
    \end{align}
    since we may approximate this integrand by
    \[
        \left( 1- \frac{t}{z + 1}\right)^{\ell - 1} = e^{-at} \left(1 - O\left(\frac{t}{z+1}\right)\right)^{at} = e^{-at} - O\left(\frac{a}{z+1}\right)
    \]
    for constant $a$ and $t \in [0,1]$. 

    We also need the probability the maximum of the arrivals in $\bX_+$ is at least $x^{(\ell')}$.
    This is
    \begin{align}
        \Pr{\MAX(\bX_+) \geq x^{(\ell')}} &= \sum_{j=1}^{\ell' } \Pr{\MAX(\bX_+) = x^{(j)}} \notag \\
        & = \sum_{j=1}^{\ell' } \Pr{\bigwedge_{j' < j} \left[j' \not\in [n] \right] \wedge \left[ j \in [n] \right]} \notag \\
        & \geq \sum_{j=1}^{\ell' } \left( \left(1- \frac{1}{z + 1}\right)^{j-1} \frac{1}{z + 1} - O\left(\frac{j}{n(z+1)} \right) \right) \notag \\
        &= \frac{1}{z + 1} \cdot \sum_{j=0}^{\ell' -1} \left(1- \frac{1}{z + 1}\right)^{j} - O\left(\frac{\ell'^2}{n(z+1)} \right) \notag \\
        & = 1- \left(1- \frac{1}{z + 1}\right)^{\ell'} - O\left(\frac{\ell'^2}{n(z+1)} \right). \label{eq:prob-max-above-ell}
        \intertext{Taking a union bound over the events that $x^{(1)}, \ldots, x^{(\ell')}$ arrive in the time interval $\bX_+$ also gives}
        \Pr{\MAX(\bX_+) \geq x^{(\ell')}} &\leq \frac{\ell'}{z+1}. \label{eq:prob-max-above-ell-UB}
    \end{align}

    In pursuit of an approximate stochastic dominance inequality along the lines of \eqref{eq:approx-stoch-dom}, for the case when $\cDp = \cD$, \cref{eq:low-l'-alg-bound,eq:prob-max-above-ell,eq:prob-max-above-ell-UB} together imply that for all $\ell' \leq \ell$, 
    \begin{align}
        \frac{\Pr{\ALG(\bX) \geq x^{(\ell')}}}{\Pr{\MAX(\bX_+) \geq x^{(\ell')}}} 
        &\geq \tau_1 \log\left( \frac{1}{\tau_1}\right) + \frac{1}{a}\left(e^{-a\tau_1} - e^{-a\tau_2}\right) -O\left(\frac{a}{z}\right)  - O\left(\frac{a^2 z}{n} \right) \notag \\
        &= \beta(\tau_1, \tau_2, z) \label{eq:prophet-prob-ratio}
    \end{align}
    for $a = 1$ as in \Cref{ass:a-value}, provided that $\frac{1}{a} \left(e^{-a\tau_1} - e^{-a\tau_2}\right) \geq O\left(\frac{a}{z+1}\right)$. This holds for fixed $\tau_1 \neq \tau_2$ and $a$ for sufficiently large $z$.
    
    This proves the claim.
\end{proof}

\prophetProbRatioBottom*

\begin{proof}[Proof of \Cref{lem:prophet-prob-ratio-bottom}]
    We start from \eqref{eq:low-l'-alg-bound} with $\ell' = \ell$, which satisfies the condition $\ell'\leq \ell$. 
    We will apply \Cref{lem:condnl-prob-high-j} and again integrate over $t \in [\tau_2, 1]$ to get probability bounds for $j > \ell$. 
    Taking $a = 1$ in accordance with \Cref{ass:a-value}, we have
    \begin{align} 
        \Pr{\ALG(\bX) \geq x^{(\ell')}}
        &= \Pr{\ALG(\bX) \geq x^{(\ell)}} + \sum_{j = \ell + 1}^{\ell'} \Pr{\ALG(\bX) = x^{(j)}} \notag \\
        &\geq \left(1 - \left(1- \frac{1}{z + 1}\right)^{\ell} \right) \cdot \tau_1 \log\left( \frac{1}{\tau_1}\right) + \frac {\ell}{z+1} \left( \left(e^{-\tau_1} - e^{-\tau_2}\right) - O\left(\frac{1}{z+1}\right)\right) \notag \\
        &\qquad\qquad + \sum_{j = \ell + 1}^{\ell'} \int_{\tau_2}^1 \left(1 - \frac{t}{z+1}\right)^{j-1} \frac{\tau_2}{t}\cdot \frac{1}{z+1} \:dt \notag \\
        &\geq \left(1 - \left(1- \frac{1}{z + 1}\right)^{\ell} \right) \cdot \tau_1 \log\left( \frac{1}{\tau_1}\right) + \frac {\ell}{z+1} \left(\left(e^{-\tau_1} - e^{-\tau_2}\right) - O\left(z^{-1}\right)\right) \notag \\
        &\qquad\qquad + \frac{1}{z+1} \cdot \sum_{j = \ell + 1}^{\ell'} \left(1 - \frac{1}{z+1}\right)^{j-1} \tau_2 \log \left(\frac{1}{\tau_2}\right) \notag \\
        &= \left(1 - \left(1- \frac{1}{z + 1}\right)^{\ell} \right) \cdot \tau_1 \log\left( \frac{1}{\tau_1}\right) + \frac {\ell}{z+1} \left(\left(e^{-\tau_1} - e^{-\tau_2}\right) - O\left(z^{-1}\right)\right) \notag \\
        &\qquad\qquad + \left(\left(1- \frac{1}{z + 1}\right)^{\ell} - \left(1- \frac{1}{z + 1}\right)^{\ell'}\right) \cdot \tau_2 \log \left(\frac{1}{\tau_2}\right) \notag \\
        &= \left(1 - \left(1- \frac{1}{z + 1}\right)^{\ell'} \right) \cdot \tau_1 \log\left( \frac{1}{\tau_1}\right) + \frac {\ell}{z+1} \left(\left(e^{-\tau_1} - e^{-\tau_2}\right) - O\left(z^{-1}\right)\right), \label{eq:large-l'-prob-bound}
    \end{align}
    where the last step followed from \Cref{eq:ass-tau1-tau2} and collecting terms.
    
    Recall that our expression \eqref{eq:prob-max-above-ell} for $\Pr{\MAX(\bX_+) \geq x^{(\ell)}}$ holds regardless of $\ell'$. We then only need an analog of \eqref{eq:prob-max-above-ell-UB} to hold here.
    But for $\ell \geq z + 1$ the bound $\Pr{\MAX(\bX_+) \geq x^{(\ell)}} \leq 1 \leq \frac{\ell}{z+1}$ holds trivially, regardless of $\ell' \geq \ell$.
    Therefore so long as $a \geq 1$ (as in \Cref{ass:a-value}), we may again upper bound $\Pr{\MAX(\bX_+) \geq x^{(\ell)}}$ by these terms.
    From \eqref{eq:large-l'-prob-bound} we then derive
    \begin{align} 
        \frac{\Pr{\ALG(\bX) \geq x^{(\ell')}}}{\Pr{\MAX(\bX_+) \geq x^{(\ell)}}}
        &\geq \tau_1 \log\left( \frac{1}{\tau_1}\right) + \left(e^{-\tau_1} - e^{-\tau_2}\right) -O\left(z^{-1}\right) \notag 
        = \beta(\tau_1,\tau_2, z), \notag
    \end{align}
    provided that $e^{-\tau_1} - e^{-\tau_2} \geq O\left(z^{-1}\right)$; for any fixed $\tau_1 \neq \tau_2$ this holds for all sufficiently large $z$. 

    This proves the claim.
\end{proof}

\algSecretaryGuarantee*

\begin{proof}[Proof of \Cref{thm:alg-secretary-guarantee}]
    We start by showing approximate stochastic dominance. 
    Let $\mu$ denote the probability measure for $\MAX(\bX_+)$. 
    By \Cref{lem:three-phase-prob-ratio-bound}, 
    \begin{align*}
        \Pr{\ALG(\bX_0, \bX_+) \geq y} &\geq \Pr{\ALG(\bX) = x_+^{(1)} \wedge x_+^{(1)} \geq y}  \\
        &= \int_{y}^\infty \Pr{\ALG(\bX) = x_+^{(1)} \given x_+^{(1)} = y'} \:d\mu(y') \\
        &\geq \alpha \cdot \int_{y}^\infty \:d\mu(y') = \alpha \cdot \Pr{\MAX(\bX_+) \geq y}.
    \end{align*}
    The competitive ratio then follows from the observation that
    \begin{align*}
        \expect{\ALG(\bX_0, \bX_+)} &= \int_{\bbR_{\geq 0}} \Pr{\ALG(\bX_0, \bX_+) \geq y} \:dy \\
        &\geq \alpha \cdot \int_{\bbR_{\geq 0}} \Pr{\MAX(\bX_+) \geq y} \:dy \\
        &= \alpha \cdot \expect{\MAX(\bX_+)}. \qedhere
    \end{align*}
\end{proof}

\threephaseProphetGuarantee*

\begin{proof}[Proof of \Cref{thm:alg-prophet-guarantee}]
    Since $\cDp = \cD$, we may consider the following alternative process for generating $\bX = \bX_0 \cup \bX_+$:
    \begin{enumerate}
        \item Sample arrival times $T = (t_1, t_2, \ldots)$ where each $ t_i \sim Unif([-z,1])$ independently. \label{item:step1}
        \item Sample a number of arrivals $k \sim Pois((z+1) \cdot n)$ and let $V = (v_1, \ldots, v_k)$ be $k$ i.i.d. draws from $\cD$. 
        \item Construct $\bX$ by placing $v^{(i)}$ at $t_i$, where $v^{(i)}$ is the $i$th largest value in $V$.
    \end{enumerate}
    Note that if $\ALG$ decides whether or not to stop based only on ordinal queries within $\bX$, then the event that $\ALG$ stops on a given $x^{(i)}$ depends only on Step \ref{item:step1}, and is independent of the value that $x^{(i)}$ takes in the subsequent steps.

    In pursuit of approximate stochastic dominance along the lines of \eqref{eq:approx-stoch-dom}, for a given $\bX$ and $y\in \bbR_{\geq 0}$, let $\lfloor y\rfloor_{\bX} \defeq \max \{X \in \bX: X \leq y\} \cup \{0\}$ be the largest value in $\bX$ that $y$ exceeds.
    Then by \Cref{lem:prophet-prob-ratio} we have
    \begin{align*}
        \Pr{\ALG(\bX_0, \bX_+) \geq y} &= \sum_i \Pr{\ALG(\bX_0, \bX_+) \geq x^{(i)} \:\wedge\: \lfloor y \rfloor_{\bX} = x^{(i)}} \\
        &= \sum_i \Pr{\ALG(\bX_0, \bX_+) \geq x^{(i)}} \cdot \Pr{\lfloor y \rfloor_{\bX} = x^{(i)}} \\
        &\geq \beta(\tau_1, \tau_2, z) \cdot \sum_i \Pr{\MAX(\bX_+) \geq x^{(i)}} \cdot \Pr{\lfloor y \rfloor_{\bX} = x^{(i)}} \\
        &= \beta(\tau_1, \tau_2, z) \cdot \sum_i \Pr{\MAX(\bX_+) \geq x^{(i)} \:\wedge\: \lfloor y \rfloor_{\bX} = x^{(i)}} \\
        &= \beta(\tau_1, \tau_2, z) \cdot \Pr{\MAX(\bX_+) \geq y}.
    \end{align*}
    
    The competitive ratio then follows as usual. 
\end{proof}  
\section{Proofs from \texorpdfstring{\Cref{sec:poisson}}{Section 4.2}} \label{app:poisson}

\begin{claim} \label{thm:hard-opt-CR}
    For the distribution defined in \eqref{eq:hard-dist-CDF}, the optimal policy has a competitive ratio of $\beta_n + \Theta(e^{-n})$. 
\end{claim}

\begin{proof}[Proof of \Cref{thm:hard-opt-CR}]
This proof proceeds in three parts.

\paragraph{Expectation of Maximum.} Since the arrival of Poisson draws from our distribution \eqref{eq:CDFdefn} with value at least $x$ is itself a Poisson process with rate $n(1-\CDF(x))$, we may write the expected maximum as
\begin{align}
    MAX_q & = \int_0^\infty \left(1 - \exp \left(-n \left(1 - \CDF(x)\right)\right)\right) dx \\
    &= \int_0^{\tilde{r}^\ast(q)} (1 - \tilde{y}((\tilde{r}^\ast)^{-1}(x)) dx + (H - \tilde{r}^\ast(q))(1 - \tilde{y}(q)) \\
    & = \int_1^q \frac{1 - \tilde{y}(t)}{\tilde{y}'(t)} dt + (H - \tilde{r}^\ast(q))(1 - \tilde{y}(q)) \label{eq:max-q-change-of-vars}\\
    & = - \int_q^1 \frac{1 - \tilde{y}(t)}{\tilde{y}'(t)} dt + (H - \tilde{r}^\ast(q))(1 - \tilde{y}(q)) \\
    & = -\int_q^1 \frac{1 - \tilde{y}(t)}{\tilde{y}'(t)} dt + \left(\frac{1}{\tilde y'(q) \ln \tilde y(q)}\right)(1 - \tilde{y}(q)),
    \intertext{where\eqref{eq:max-q-change-of-vars} follows by substitution with $x = \tilde{r}^*(t)$, together with the fact that $(\tilde{r}^*)'(t) = 1/\ty'(t)$ by \eqref{eq:r-defn}. Taking the limit then gives}
    \lim_{q \rightarrow 0} MAX_q &= -\int_0^1 \frac{1 - \tilde{y}(t)}{\tilde{y}'(t)} dt - \frac{1}{\tilde{y}'(0)}, \label{eq:max-form-no-q}
\end{align}
because by l'H\^{o}pital's rule
\[
    \lim_{q \to 0} \frac{1 - \tilde{y}(q)}{\ln \tilde{y}(q)} = \lim_{q \to 0} \frac{- \tilde{y}'(q)}{\frac{\tilde{y}'(q)}{\tilde{y}(q)}} = - \tilde{y}(0) = - 1.
\]

\paragraph{Performance of Optimal Policy.} By our definition of $\tilde{r}^\ast$, we have for any $q$ if $t > q$ then
\begin{align*}
\frac{d}{dt} \tilde{r}^\ast(t) & = \frac{1}{\tilde{y}'(t)} \\
& = - \int_q^t \frac{\tilde{y}''(t)}{(\tilde{y}'(t))^2} dt + \frac{1}{\tilde{y}'(q)} \\
& \stackrel{(\star)}{=} - \int_q^t \frac{\ln \tilde{y}(t)}{\tilde{y}'(t)} dt + \frac{1}{\tilde{y}'(q)}\\
& \stackrel{(\star\star)}{=} - \int_{q}^{t} n \cdot (1 - \CDF(\tilde{r}^\ast(t))) \cdot (\tilde{r}^\ast)'(t) dt + (H - \tilde{r}^\ast(q)) \ln \tilde{y}(q)  \\
& = - \int_{\tilde{r}^\ast(t)}^{\tilde{r}^\ast(q)} n \cdot (1 - \CDF(x)) dx - \int_{\tilde{r}^\ast(q)}^\infty n \cdot (1 - \CDF(x)) dx \\
& = - \int_{\tilde{r}^\ast(t)}^\infty n \cdot (1 - \CDF(x)) dx, 
\end{align*}
where in step $(\star)$, we use that $\tilde{y}''(t) = \tilde{y}'(t) \ln(\tilde{y}(t))$ and, in step $(\star\star)$, we use the definition of $H$.

So, $\tilde{r}^\ast$ defines an optimal policy on $[q, 1]$.

Before time $q$, the optimal policy accepts any $H$ value that arrives. These $H$ values arrive at rate $-n \left( - \frac{1}{n} \ln \tilde{y}(q) \right) = - \ln \tilde{y}(q)$. So, the probability of none of them arriving by time $q$ is $\exp\left(q \ln \tilde{y}(q))\right) = (\tilde{y}(q))^q$. Therefore, the optimal policy's expected reward is
\begin{align}
    OPT_q & = \left(1 - (\tilde{y}(q))^q \right) H + (\tilde{y}(q))^q \tilde{r}^\ast(q) \\
    &= \left(1 - (\tilde{y}(q))^q \right) (H - \tilde{r}^\ast(q)) + \tilde{r}^\ast(q) \\
    & = \left(1 - (\tilde{y}(q))^q \right) \frac{1}{\tilde y'(q) \ln \tilde y(q)} + \tilde{r}^\ast(q) \\
    \lim_{q \rightarrow 0} OPT_q &= \tilde{r}^\ast(0) = - \int_0^1 \frac{1}{\tilde{y}'(s)} ds, \label{eq:opt-form-no-q}
\end{align}
where we use that
\[
\lim_{q \to 0} \frac{1 - (\tilde{y}(q))^q}{\ln \tilde{y}(q)} = \lim_{q \to 0} \frac{- (\tilde{y}(q))^q \left( \ln(\tilde{y}(q)) + q \frac{\tilde{y}'(q)}{\tilde{y}(q)} \right)}{\frac{\tilde{y}'(q)}{\tilde{y}(q)}} = 0
\]
by l'Hôpital's rule, together with the observation that by \eqref{eq:y-tilde-defn} $\tilde{y}'(0)$ is a constant. 

\paragraph{Competitive Ratio.} 
Combining \eqref{eq:max-form-no-q} and \eqref{eq:opt-form-no-q}, we have
\begin{equation}
    \inf_{q > 0} \frac{OPT_q}{MAX_q} \leq \lim_{q \to 0} \frac{OPT_q}{MAX_q} =  \frac{\lim_{q \to 0} OPT_q}{\lim_{q \to 0} MAX_q} = \frac{\int_0^1 \frac{1}{\tilde{y}'(t)} dt}{\frac{1}{\tilde{y}'(0)} + \int_0^1 \frac{1 - \tilde{y}(t)}{\tilde{y}'(t)} dt}. \label{eq:Rdefn}
\end{equation}

We claim this ratio is nearly $\beta_n$, approaching it in the limit as $n$ becomes large. 
In particular, let $R := \lim_{q \rightarrow 0} \frac{OPT_q}{MAX_q}$, and let $\eps := 1/R - 1/\beta_n$, so that $1/R = 1/\beta_n + \eps$. This defines $\eps$. We will now get a handle on it.

As $\tilde{y}'(0) = -1/\beta_n$ by the definition of $\ty$ \eqref{eq:y-tilde-defn}, rearranging \eqref{eq:Rdefn} yields
\begin{align}
    \int_0^1 \frac{1}{\tilde{y}'(t)} \left( \left(1 - \frac{1}{R} \right) - \tilde{y}(t) - \beta_n \tilde{y}'(t) \right) dt &= 0 \notag \\
    -\int_0^1 \frac{1}{\tilde{y}'(t)} \eps \: dt + \int_0^1 \frac{1}{\tilde{y}'(t)} \left( \left(1 - \frac{1}{\beta_n} \right) - \tilde{y}(t) - \beta_n \tilde{y}'(t) \right) dt &= 0 \\
    \left(\int_0^1 \frac{dt}{\tilde{y}'(t)} \right)^{-1} \cdot \int_0^1 \frac{1}{\tilde{y}'(t)} \left( \left(1 - \frac{1}{\beta_n} \right) - \tilde{y}(t) - \beta_n \tilde{y}'(t) \right) dt &= \eps. \label{eq:rearranged-ratio}
\end{align}
Again applying \eqref{eq:y-tilde-defn} gives
\begin{align}
    \eps &= \left(\int_0^1 \frac{dt}{\tilde{y}'(t)} \right)^{-1} \cdot \int_0^1 \frac{1}{\tilde{y}'(t)} \left( - \ty(t) \ln \ty(t) + \left(1-\beta_n\right) \ty'(t) \right) dt \notag \\
    &= \left(\int_0^1 \frac{dt}{\tilde{y}'(t)} \right)^{-1} \cdot \left( - \int_0^1 \frac{\ty(t) \ln \ty(t)}{\tilde{y}'(t)} dt + \left(1-\beta_n\right)  \right) \notag \\
    &= \left(\int_0^1 \frac{dt}{\tilde{y}'(t)} \right)^{-1} \frac{1}{e^{n}\cdot (1/\beta_n - 1) - (n+1)}, \label{eq:eps-form}
\end{align}
where this last step uses the observation of \cite{liu21variable}) that by \eqref{eq:y-tilde-defn} and integration by parts,
\[
    \int_0^1 \frac{\ty(t) \log \ty(t)}{\tilde{y}'(t)} dt = \int_0^1 \frac{\ty''(t)  \ty(t)}{(\ty'(t))^2} dt = - \left.\frac{\ty(t)}{\ty'(t)}\right\vert_0^1 + \int_0^1 dt = 1 - \beta_n + \frac{e^{-n}}{e^{-n}(n+1) - (1/\beta_n - 1)}.
\]
Since $-\beta_n^{-1} + 1 \leq \ty'(t) \leq -\beta_n^{-1}$ for all sufficiently large $n$, this first definite integral term is a constant.

\medskip
We conclude that $\eps = \Theta(e^{-n})$ and the competitive ratio is $R = \beta_n + \Theta(e^{-n})$, as claimed. 
\end{proof}

In particular, this shows that the optimal competitive ratio approaches the Kertz constant $\beta$ at least exponentially quickly in the Poisson rate.

To prove this, we only need to understand the rate at which $\beta_n$ as defined in \eqref{eq:bn-defn} approaches $\beta$.

\begin{lemma} \label{lem:modifiedbetabound}
    For $t \in [0,1/4]$ let $\beta_t$ be the constant for which
    \[
        \int_{t}^1 \frac{1}{\beta_t^{-1} - 1 - y(\log y - 1)} \: dy = 1,
    \]
    where $\beta_0 = \beta=0.745...$ is the Kertz constant. Then $\beta_t$ is well-defined, and $\beta_t = \beta + O(t)$.
\end{lemma}

\begin{proof}%[Proof of \Cref{lem:modifiedbetabound}]
    For convenience let $c^* = \beta^{-1}-1$ (corresponding to the Kertz constant) and $c_t = \beta_t^{-1}-1$, and define $F(t, c)$ to be
    \[
        F(t, c) := \int_{t}^1 \frac{1}{c - y(\log y - 1)} \: dy,
    \]
    for $t \in [0,1]$ and $c > 0$. Then $c_t$ is defined to be the $c$ for which $F(t,c) = 1$, and $F(0, c^*) = 1$.

    To begin observe that $F(t,c)$ is decreasing in $t$ for fixed $c$, and decreasing in $c$ for fixed $t$. Then $F(t,c) \leq F(0,c^*)=1$ and $F(t,c_t)=1$ imply that $c_t \leq c^*$. Our goal is therefore to establish the bound $c_t \geq c^* - O(t)$, which will imply the claim. We will do this by demonstrating some appropriate $c'$ for which $F(t,c') \geq 1$; then $c_t \geq c'$ since $F(t,c')$ is decreasing in $c'$.
    
    To this end, consider the derivative of $F(0,c)$ with respect to $c$ (which is negative):
    \begin{align}
        \frac{\partial}{\partial c} F(0,c) 
        &= \frac{\partial}{\partial c} \int_{0}^1 \frac{1}{c - y(\log y - 1)} \: dy \notag \\
        &= \int_{0}^1 \frac{\partial}{\partial c} \frac{1}{c - y(\log y - 1)} \: dy \notag \\
        &= \int_{0}^1 \frac{-1}{(c - y(\log y - 1))^2} \: dy \notag \\
        &\leq \frac{-1}{c^2} \notag
    \end{align}
    by the Leibniz integral rule and using that $-y \log y + y \geq 0$. This implies that for $c' \leq c$,
    \begin{equation} \label{eq:betaintegralderivativelb}
        F(0,c') \geq F(0,c) + (c' - c)\cdot\frac{-1}{c^2}.
    \end{equation}
    
    Next consider the first part of the integral:
    \begin{align}
        F(0,c) - F(t,c) &= \int_{0}^t \frac{1}{c - y(\log y - 1)} \: dy \notag \\
        &\leq \int_{0}^t \frac{1}{c} \: dy = t/c, \label{eq:betaintegralprefixbound}
    \end{align}
    because $-y \log y + y$ is nonnegative for $y \in [0,1]$.

    Combining \eqref{eq:betaintegralderivativelb} and \eqref{eq:betaintegralprefixbound} we have
    \begin{align}
        F(t,c') &= F(0,c') - (F(0,c') - F(t,c')) \notag \\
        & \geq F(0,c^*) + \frac{c^* - c'}{(c^*)^2} - \frac{t}{c'}. \notag
    \end{align}
    Since $F(0,c^*) = 1$, this is true so long as $\frac{c^* - c'}{(c^*)^2} \geq \frac{t}{c'}$, which holds for $c' = c^*(1-2t)$ so long as $t \leq 1/4$. It follows that for all such $t$ we have $c_t \geq c^* (1-2t)$. In this case 
    \[
        \beta_t \leq \beta' = \frac{1}{(\beta^{-1} - 1)(1 - 2t) +1} = \beta + O(t),
    \]
    as claimed.
\end{proof}

We are now ready to prove \Cref{thm:kertz-hardness-n-asymptotic}.
 
\begin{proof}[Proof of \Cref{thm:kertz-hardness-n-asymptotic}]
    This follows directly from \Cref{thm:hard-opt-CR}, together with the definition of $\beta_n$ \eqref{eq:bn-defn} and \Cref{lem:modifiedbetabound}, which demonstrates that $\beta_n = \beta + O(e^{-n})$.
\end{proof}
\section{Proofs from \texorpdfstring{\Cref{sec:hardness}}{Section 4}} \label{app:hardness}

\optimalpolicystructure* 

\begin{proof}[Proof of \Cref{lem:opt-policy}]
    Let $\epsilon > 0$ and consider $r(t) - r(t+\epsilon)$. Let $Z$ be the value of the highest arrival in $[t, t+\epsilon)$; set $Z = 0$ if there is no arrival. Note that $r(t) \leq \Ex{\max\{r(t+\epsilon), Z\}} = r(t+\epsilon) + \Ex{\max\{0, r(t+\epsilon)\}}$.

    As we have $\Pr{Z \geq x} = n \cdot \epsilon \cdot (1 - \CDF(x))$, this gives us
    \[
    r(t) - r(t+\epsilon) \leq \int_{r(t+\epsilon)}^\infty n \cdot \epsilon \cdot (1 - \CDF(x)) dx.
    \]
    Taking the limit $\epsilon \to 0$, we get the first claim.

    Furthermore, accepting any arrival above $r(t + \epsilon)$, we obtain
    \[
    r(t) - r(t+\epsilon) \geq \int_{r(t+\epsilon)}^\infty (1 - \exp(-n\epsilon) \cdot (1 - \CDF(x)) dx.
    \]
    Taking the limit $\epsilon \to 0$, this gives $r'(t) \geq \int_{r(t)}^\infty n \cdot (1 - \CDF(x)) dx$.
\end{proof}

\lemprophimpossprobabilities*
\begin{proof}[Proof of \Cref{lem:proph-imposs-probabilities}]
    Let
    \begin{itemize}
    \item $\mathcal{E}_1$ be the event that we stop on the first value of at least $b$ until time $\epsilon$,
    \item $\mathcal{E}_2$ be the event that we stop on a value of at least $b$ until time $\epsilon$ that is not the first one,
    \item $\mathcal{E}_3$ be the event that we stop on a value below $b$ but at least $c$ until time $\epsilon$,
    \item $\mathcal{E}_1'$ be the event that there is exactly one arrival of value at least $b$ until time $\epsilon$, and
    \item $\mathcal{E}_2'$ be the event that there are two or more of them. 
    \end{itemize}

    Note that for one of the events $\mathcal{E}_1, \mathcal{E}_2, \mathcal{E}_3$ to take place it is necessary that there is an arrival of value at least $c$ until time $\epsilon$. Therefore, $\Pr{\mathcal{E}_1 \cup \mathcal{E}_2 \cup \mathcal{E}_3} \leq 1 - e^{-n \epsilon (1 - \CDF(c))}$.

    The arrivals above the value $b$ are themselves a Poisson process with rate $\lambda = n (1 - \CDF(b))$. So, letting $s(t)$ denote the probability that the policy has not stopped by time $t$, the overall probability of accepting any arrival above $b$ up to time $\epsilon$ is given by
    \begin{align*}
    \Pr{\mathcal{E}_1 \cup \mathcal{E}_2} & = \int_0^\epsilon s(t) \cdot \lambda \cdot (1 - p(t)) dt \geq s(\epsilon) \cdot \lambda \cdot \left( \epsilon - \int_0^\epsilon p(t) dt \right) \\
    & = (1 - \Pr{\mathcal{E}_1 \cup \mathcal{E}_2 \cup \mathcal{E}_3}) \cdot \lambda \cdot ( \epsilon - P ) \\
    & \geq e^{-n \epsilon (1 - \CDF(c))} \cdot \lambda \cdot ( \epsilon - P ).
    \end{align*}
    
    Our assumption on the policy is that $\Pr{\mathcal{E}_1 \given \mathcal{E}_1' \cup \mathcal{E}_2'} \leq (1-\delta)$.
    And since $\mathcal{E}_1 \subseteq \mathcal{E}_1' \cup \mathcal{E}_2'$, this implies that $\Pr{\mathcal{E}_1} \leq (1 - \delta) \Pr{\mathcal{E}_1' \cup \mathcal{E}_2'}$. 
    
    The number of arrivals above $b$ until time $\epsilon$ is drawn from a Poisson distribution with rate $\lambda' \defeq - \eps \lambda$ and so we have
    \begin{align*}
        \Pr{\mathcal{E}_1'} &= \lambda' e^{- \lambda'}, \\
        \Pr{\mathcal{E}_2'} &= 1 - e^{- \lambda'} - \lambda' e^{-\lambda'}, \\
        \Pr{\mathcal{E}_1' \cup \mathcal{E}_2'} &= 1 - e^{- \lambda'}. 
    \end{align*}
    By concavity we have $\frac{1 - e^{- \lambda'} - \lambda' e^{- \lambda'}}{1 - e^{- \lambda'}} \leq \frac{1}{2} \lambda'$, and so $\Pr{\mathcal{E}_2} \leq \Pr{\mathcal{E}_2'} \leq \frac{1}{2} \lambda' \Pr{\mathcal{E}_1' \cup \mathcal{E}_2'} $.
    Therefore 
    \begin{align*}
        \Pr{\mathcal{E}_1 \cup \mathcal{E}_2} &\leq  \Pr{\mathcal{E}_1} + \Pr{\mathcal{E}_2} \\
        &\leq (1 - \delta) \Pr{\mathcal{E}_1' \cup \mathcal{E}_2'} + \frac{1}{2} \lambda' \Pr{\mathcal{E}_1' \cup \mathcal{E}_2'} \\
        &\leq \left(1 - \delta + \frac{1}{2} \lambda' \right) \cdot \Pr{\mathcal{E}_1' \cup \mathcal{E}_2'} \\
        &\leq \left(1 - \delta + \frac{1}{2} \lambda' \right) \cdot \lambda'.
    \end{align*}

    In combination, we have
    \begin{align*}
    P \geq \epsilon - \frac{1}{\lambda \cdot e^{-n \epsilon (1 - \CDF(c))}} \Pr{\mathcal{E}_1 \cup \mathcal{E}_2} = \epsilon \left( 1 - \frac{1 - \delta + \frac{1}{2} \epsilon \lambda}{e^{-n \epsilon (1 - \CDF(c))}} \right). \qedhere
    \end{align*}
\end{proof}

\obstildeyderivativebounds*
\begin{proof}[Proof of \Cref{obs:tilde-y-derivative-bounds}]
    This follows from the differential equation \eqref{eq:y-tilde-defn} that defines $\tilde{y}$ and the observation regarding it.
    From these prior observations and the definition of $\beta_n$, we know that $\tilde{y}(t)$ is decreasing and continuous, and ranges from $\tilde{y}(0)=1$ to $\tilde{y}(1)=e^{-n}$.

    Next consider the function $y (\log y - 1)$, which appears in the first part of \eqref{eq:y-tilde-defn}.
    For $y \in [0,1]$, it attains a maximum of $0$ at $y=0$ and a minimum of $-1$ at $y = 1$.
    From \eqref{eq:y-tilde-defn} we then conclude that 
    \begin{equation*}
        - \beta_n^{-1} \leq \tilde{y}'(t) \leq - \left( \beta_n^{-1} - 1\right).
    \end{equation*} 
    Since $\tilde{y}$ is differentiable, integrating starting from $t=0$ implies the first part of the second claim.
    To get the last part, by Taylor's theorem we have
    \[
        \ty(t) \leq 1 + t\cdot \ty'(0) + \frac{t^2}{2} \cdot \max_{y(s) \in [1/e, 1]} \ty''(s), \qquad\qquad t \leq \ty^{-1}(1/e).
    \]
    The claim follows by observing that $\ty'' = \ty' \log \ty$ is decreasing in $y$, and furnishing an upper bound at $\ty = 1/e$.
\end{proof}

\lemprophetmuststopearly*
\begin{proof}[Proof of \Cref{lem:prophet-must-stop-early}]
    Consider any policy that fulfills the constraint from the lemma, and consider its performance on the hard instance \eqref{eq:hard-dist-CDF}. Without loss of generality, we can assume that it does not stop on values below $\tilde{r}^\ast(\epsilon)$ before time $\epsilon$. Otherwise, modify the policy so that it rejects such a value it wanted to accept, waits until time $\epsilon$, and then starts an optimal (unconstrained) policy. The expected reward of this is at least $\tilde{r}^\ast(\epsilon)$. So, overall, the expected reward of our constrained policy can only have increased.

    We will first derive a lower bound on $P$ in terms of $\eps$ and $\delta$.
    We start from \Cref{lem:proph-imposs-probabilities} with the choice $c = \tilde{r}^\ast(\epsilon)$, and apply that $1-\CDF(\rts(t)) = -\frac{1}{n}\ln(\ty(t))$ by \eqref{eq:CDFdefn}.
    Then for $t_b \defeq (\rts)^{-1}(b)$, 
    \begin{align}
        P &\geq \epsilon \left( 1 - \frac{1 - \delta + \frac{1}{2} \epsilon n (1 - \CDF(b))}{e^{-n \epsilon (1 - \CDF(c))}} \right) \notag \\ 
        &= \epsilon \left( 1 - \frac{1 - \delta + \frac{1}{2} \epsilon  \ln \frac{1}{\ty(t_b)}}{(\ty(\eps))^{\eps}} \right) \notag \\
        &\geq \eps \left( 1 - \left(1-\frac{1}{\beta_n} \eps\right)^{-\eps} \left(1 - \delta + \frac{1}{2} \epsilon  \ln \frac{1}{\ty(\eps)}\right) \right) \notag \\
        &\geq \eps \left( 1 - \left(1-\frac{1}{\beta_n} \eps\right)^{-\eps} \left(1 - \delta + \frac{1}{2} \eps \ln \frac{1}{1 - \frac{1}{\beta_n} \eps }\right) \right), \notag \\
        &= \eps \left( \delta  - (1 - \delta)\left(\left(1-\frac{1}{\beta_n} \eps\right)^{-\eps} - 1 \right) - \left(1-\frac{1}{\beta_n} \eps\right)^{-\eps} \frac{1}{2} \eps \ln \frac{1}{1 - \frac{1}{\beta_n} \eps }\right) . \notag 
        \intertext{
            We next apply two observations. First is that $\left(1-\frac{1}{\beta_n} \eps\right)^{-\eps} \leq 1 + \frac{3}{2} \eps^2$ for $\eps \leq \frac{1}{8}$ and all sufficiently large $n$. The second is that this second term is bounded by $\left(1-\frac{1}{\beta_n} \eps\right)^{-\eps} \frac{1}{2} \eps \ln \frac{1}{1 - \frac{1}{\beta_n} \eps } \leq \frac{3}{4} \eps^2 $ for all $\eps \leq \frac{1}{8}$ and all sufficiently large $n$. Then
        }
        &\geq \eps \left( \delta  - (1 - \delta)\frac{3}{2} \eps^2 - \frac{3}{4} \eps^2 \right) \notag \\
        &= \eps \delta \left(1 -\frac{3}{2} \eps^2\right) - \frac{9}{4} \eps^3.
        \label{eq:hardness-prophet-p-lower-bound}
    \end{align}
    This concludes our lower bound on $P$.
    \medskip

    We now upper bound the reward of the policy in terms of $P$.
    By \Cref{lem:proph-imposs-CR-bound}, the reward of the policy is no more than $\max\left\{ a, \tilde{r}^\ast(\frac{\kappa P (1-\gamma)}{2 (1 + \kappa)})\right\}$, where $a = \tilde{r}^\ast(\gamma P)$ and
    \begin{align*}
        \kappa &= \frac{(b - a) \cdot (1 - \CDF(b))}{\int_{\rts(\eps)}^\infty (1-\CDF(x)) dx} \\
        &= (b - a) \frac{(-\ln(\tilde{y}((\tilde{r}^\ast)^{-1}(b)))) } {\frac{-1}{\tilde{y}'(\epsilon)}} \\
        \intertext{by \eqref{eq:r-defn} and again by \eqref{eq:CDFdefn}. Then by \Cref{obs:tilde-y-derivative-bounds},}
        &\geq(b - a) \beta_n \cdot (-\ln(\tilde{y}((\tilde{r}^\ast)^{-1}(b)))).
        \intertext{We will choose $\gamma = 2/3$ so that $a = \rts\left(\frac{2}{3} P\right)$, and let $b = \rts\left(\frac{1}{3} P \right)$. Then by \Cref{obs:tilde-y-derivative-bounds} we have}
        &= (b - a) \beta_n \cdot \ln\frac{1}{\tilde{y}(P/3)}  \\
        &\geq (b - a) \beta_n \ln \left(1 - \frac{1}{\beta_n} \cdot \frac{P}{3} + \frac{1}{2}\frac{P^2}{9} \cdot \left(1 + \frac{2}{e} - \frac{1}{\beta_n}\right)\right)^{-1} \\
        &\geq (b - a) \frac{P}{3}.
    \end{align*}

    We will now apply this lower bound on $\kappa$.
    Observing that $\frac{\kappa}{1 + \kappa}$ is increasing in $\kappa$ and that $(b-a)P/3 \leq 1/24$ for $\eps \leq 1/8$, this gives
    \begin{align*}
        \frac{\kappa P (1-\gamma)}{2 (1 + \kappa)} &= \frac{P}{6} \cdot \frac{\kappa }{1 + \kappa} \\
        &\geq \frac{P}{6} \cdot \frac{(b - a) \frac{P}{3} }{1 + (b - a) \frac{P}{3}} \\
        &\geq \frac{P^2}{18} \cdot \frac{19}{20}(b-a),
    \end{align*}
    where we applied that $\frac{x}{1+x} \geq \frac{19}{20}x$ for $x\in [0,1/24]$.

    We now bound $\rts$. By \eqref{eq:r-defn} and \Cref{obs:tilde-y-derivative-bounds} we have
    \begin{align}
        \rts(t) &\defeq \int_t^1 \frac{-1}{\ty'(s)} \: ds 
        = \rts(0) - \int_0^t \frac{-1}{\ty'(s)} \: ds \notag\\
        &= \beta_n - \int_0^t \frac{-1}{\ty'(s)} \: ds \notag\\
        &\leq \beta_n - \int_0^t \frac{1}{\beta_n^{-1}} \: ds \notag\\
        &=\beta_n (1 - t). \label{eq:rts-t-upper-bound}
    \end{align}
    
    Substituting, we then have from \Cref{lem:proph-imposs-CR-bound} that one component of our reward upper bound is
    \begin{align}
        a & = \rts(2P/3) \notag \\
        & \leq \beta_n - \beta_n \cdot \frac{2P}{3} .
        \label{eq:hardness-prophet-a-upperbound}
    \end{align}
    On the other hand, by \eqref{eq:rts-t-upper-bound} the other component of our reward upper bound is
    \begin{align}
        \rts\left(\frac{\kappa P (1-\gamma)}{2 (1 + \kappa)}\right) &\leq \beta_n - \beta_n \cdot \left(\frac{\kappa P (1-\gamma)}{2 (1 + \kappa)}\right) \notag \\
        &\leq \beta_n - \beta_n \cdot \frac{P^2}{18} \cdot \frac{19}{20}(b-a). \label{eq:hardness-prophet-second-reward-bound}
    \end{align}

    It remains only to lower bound $(b-a)$ for our choice of $a = \rts(2P/3)$ and $b = \rts(P/3)$.
    From \eqref{eq:rts-t-upper-bound} we have an upper bound on $a$, so for a combined bound, we again employ \Cref{obs:tilde-y-derivative-bounds} to obtain
    \begin{align}
        \rts(t) &\defeq \int_t^1 \frac{-1}{\ty'(s)} \: ds 
        = \rts(0) - \int_0^t \frac{-1}{\ty'(s)} \: ds \notag\\
        &= \beta_n - \int_0^t \frac{-1}{\ty'(s)} \: ds \notag\\
        &\geq \beta_n - t \frac{1}{\frac{1}{\beta_n} - t \cdot \left(1 + \frac{2}{e} - \frac{1}{\beta_n}\right)}. \label{eq:rts-t-lower-bound}
    \end{align}
    Combining, we have 
    \begin{align}
        \rts(t/2) - \rts(t) &\geq \beta_n - \frac{t}{2} \frac{1}{\frac{1}{\beta_n} - \frac{t}{2} \cdot \left(1 + \frac{2}{e} - \frac{1}{\beta_n}\right)} - \beta_n (1 - t) \notag \\
        &=\beta_n \cdot t \cdot \left(1 - \frac{1}{2} \left(1 - \frac{t}{2} \frac{1 + \frac{2}{e} - \frac{1}{\beta_n}}{1/\beta_n}\right)^{-1} \right) \notag \\
        &\geq \beta_n\cdot t \cdot \frac{9}{20}, \notag 
    \end{align}
    for all $t \leq 1/2$ and sufficiently large $n$. 
    Applying this to \eqref{eq:hardness-prophet-second-reward-bound} with $t = 2P/3$ and using our prior bound on $P$ of \eqref{eq:hardness-prophet-p-lower-bound}, our second part of the reward upper bound is
    \begin{align}
        \rts\left(\frac{\kappa P (1-\gamma)}{2 (1 + \kappa)}\right) &\leq  \beta_n - \beta_n \cdot P^3 \cdot \beta_n \frac{19}{800}
        \notag \\
        &\leq  \beta_n - \frac{P^3}{76} \notag
    \end{align}
    for all sufficiently large $n$.
    The larger of this bound and \eqref{eq:hardness-prophet-a-upperbound} are an upper bound on our reward by \Cref{lem:proph-imposs-CR-bound}; however this bound is clearly the larger.

    \medskip
    It finally remains to substitute our lower bound \eqref{eq:hardness-prophet-p-lower-bound} for $P$ as a function of $\eps$ and $\delta$ into this expression. 
    Doing so yields the claim.
\end{proof}

\end{document}